\documentclass[12pt]{article}

\usepackage{eurosym}
\usepackage{graphicx}
\usepackage{enumerate}
\usepackage{amsmath}
\usepackage{amsfonts}
\usepackage{amssymb}
\usepackage{amsfonts}
\usepackage{amssymb,amsmath,amsthm,epsfig,graphicx,natbib,subfigure,hyperref}
\usepackage{geometry}
\usepackage{caption}
\usepackage{color}
\usepackage{enumerate}
\usepackage{setspace}
\usepackage[textwidth=30mm]{todonotes}
\usepackage{sgame, tikz}
\usepackage{multirow,array}
\usepackage{rotating}
\usepackage{graphicx}
\usepackage{subfigure}
\usepackage{float}
\usepackage{verbatim}
\usepackage{bbm}
\usepackage{natbib}
\usepackage{tikz}
\usetikzlibrary{calc,arrows.meta,positioning}

\definecolor{darkblue}{rgb}{0.0, 0.0, 0.55}
\hypersetup{colorlinks,linkcolor={darkblue},citecolor={darkblue},urlcolor={darkblue}} 
\allowdisplaybreaks

\bibpunct[: ]{(}{)}{;}{a}{}{,}
\newtheorem{theorem}{Theorem}
\newtheorem*{theorem*}{Theorem}
\newtheorem{lemma}{Lemma}

\newtheorem{ass}{Assumption}

\newtheorem{definition}{Definition}
\newtheorem{example}{Example}
\newtheorem*{example*}{Example}

\newtheorem{proposition}{Proposition}
\newtheorem{remark}{Remark}

\newcommand {\be}{\begin{equation}}
\newcommand {\ee}{\end{equation}}

\title{Privacy-Constrained Signals}
\author{Zhang Xu\thanks{School of Economics, Renmin University of China. \textit{\ Email}: \href{mailto:xuzhang@ruc.edu.cn}{xuzhang@ruc.edu.cn}} \and Wei Zhao\thanks{School of Economics and Management, Tsinghua University. \textit{\ Email}: \href{mailto:wei.zhao@outlook.fr}{wei.zhao@outlook.fr}}}

\begin{document}

\maketitle

\begin{abstract}
    This paper provides a unified approach to characterize the set of all feasible signals subject to privacy constraints. The Blackwell frontier of feasible signals can be decomposed into minimum informative signals achieving the Blackwell frontier of privacy variables, and conditionally privacy-preserving signals. A complete characterization of the minimum informative signals is then provided. We apply the framework to ex-post privacy (including differential and inferential privacy) and to constraints on posterior means of arbitrary statistics.
    
\end{abstract}

\newpage
\tableofcontents
\newpage

\section{Introduction}
    The big data plays a critical role in economic decisions, promoting efficiency in allocation. At the same time, growing concern on privacy has been drawn.  The abuse of sensitive (personal) data, leading to statistical and price discrimination, imposes negative externality on the economy as a whole. A natural question is, what is the set of feasible datasets, subject to privacy constraints on sensitive information? The past literature on this question mainly focus on perfectly privacy-preserving constraints. However, privacy-preserving constraints may be too demanding in practical operations. In this spirit, various orders, both complete and partial, have been proposed to measure the degree of information leakage.

    This paper develops a general framework for privacy constraints. We model information disclosure through signals defined on an abstract state of the world $\tilde{\omega} \in \Omega$, and represent the sensitive information as a random variable $\tilde{\theta} : \Omega \to \Theta$ defined on the same state space. For instance, if $\Omega$ is an $n$-dimensional space, the sensitive component $\Theta$ may correspond to its first $m$ dimensions with $m < n$. Following \cite{blackwell1953equivalent}, each signal (experiment) can be represented by the distribution of posteriors it induces. We therefore model privacy constraints as a subset of distributions over posteriors about the sensitive variable $\tilde{\theta}$, which we refer to as the \emph{privacy-permissible set}. A signal is privacy-constrained if the posterior distribution it induces, when marginalized over $\Theta$, belongs to this permissible set. For the analysis to be well behaved and  natural, we assume that the privacy-permissible set is a lower set with respect to the Blackwell order, that is, whenever a signal is permissible, any less informative (in Blackwell sense) signal is also permissible.

    To characterize the set of privacy-constrained signals, it is equivalent to describe its Blackwell frontier, that is, the set of all privacy-constrained signals that are Blackwell-undominated.
    Theorem~\ref{thm:decomposition} reduces this task to characterizing the Blackwell frontier of the privacy-permissible set itself.
    Given a distribution over posteriors about the sensitive variable, $\gamma \in \Delta(\Delta(\Theta))$, we first construct a minimum-informative extension $\tau_\gamma$, which preserves the marginal distribution over $\Delta(\Theta)$ while revealing as little information as possible about the state $\tilde{\omega}$.
    Afterward, we disclose $\tau_\gamma$, and then disclose a Blackwell-undominated signal among those that are conditionally privacy-preserving given $\tau_\gamma$. The latter class of signals is characterized in \cite{strack2024privacy}.
    This sequential procedure yields an element of the Blackwell frontier of privacy-constrained signals.

    Theorem~\ref{prp:least-informative-signal} describes how to generate all minimum-informative extensions for a given $\gamma \in \Delta(\Delta(\Theta))$. Under a minimum-informative extension, each posterior over $\tilde{\theta}$ in the support of $\gamma$ is extended to exactly one posterior over the full state $\tilde{\omega}$. Thus, computing a minimum-informative extension amounts to assigning a conditional distribution of $\tilde{\omega}$ given $\tilde{\theta}$ to every element in the support of $\gamma$. In the discrete setting, this assignment can be expressed through a finite collection of linear constraints.

    However, characterizing the Blackwell frontier of an abstract privacy-permissible set remains challenging. In Section~\ref{sec:privacy_permissible_set}, we focus on two important classes of privacy constraints and show how to characterize their Blackwell frontiers. The first class is \emph{ex-post privacy}, which encompasses differential privacy \cite{dwork2006calibrating}, inferential privacy \cite{ghosh2016inferential, wang2025inferentially}, and Bayesian privacy \cite{eilat2021bayesian}. When the regulator is concerned with ex post privacy loss, the constraint is imposed directly on the posterior beliefs about the sensitive variable $\tilde{\theta}$. Ex-post privacy specifies a subset of posteriors over $\tilde{\theta}$, and a signal is ex post privacy-constrained if every realized posterior about $\tilde{\theta}$ lies in this subset. Proposition~\ref{prp:extreme_points} shows that under ex-post privacy, characterizing the Blackwell frontier of the privacy-permissible set reduces to identifying the extreme points of the permissible posteriors over $\tilde{\theta}$. Proposition~\ref{prp:ip-froniter} further provides an explicit characterization of the Blackwell frontier under inferential privacy. Related results for the discrete cases of inferential privacy and differential privacy appear in \cite{xu2025privacy}. The second class is \emph{posterior-mean privacy}, where the regulator is concerned only with the information revealed about the posterior mean of a statistic defined on the sensitive variable. In this case, Proposition~\ref{prp:expectation} provides a clean characterization of the corresponding Blackwell frontier.

    The remainder of the paper is organized as follows. Section~\ref{sec:model} introduces the formal setting. Section~\ref{sec:pcs} presents the general characterization of privacy-constrained signals. Section~\ref{sec:privacy_permissible_set} characterizes the Blackwell frontier of the privacy-permissible set for several privacy constraints. Section~\ref{sec:discussion} offers further discussion. All proofs are collected in the Appendix.

\section{Model}\label{sec:model}

Let $(\Omega, \mathcal{B}(\Omega), \mu_0)$ be a probability space, where $\mathcal{B}(\cdot)$ is the Borel $\sigma$-algebra generator, $(\Omega, \mathcal{B}(\Omega))$ is a standard Borel space and $\mu_0 \in \Delta(\Omega)$ is an interior prior.\footnote{A Borel space $(E,\mathcal{B}(E))$ is standard if there is an isomorphism $\psi: E \to F$ for some  $F \in \mathcal{B}(\mathbb{R})$. An isomorphism is a bijection $\psi$ such that both $\psi$ and $\psi^{-1}$ are measurable. (\cite{ccinlar2011probability}, p.11.)}
The \emph{state} is a random variable $\tilde{\omega} \sim \mu_0$ and $\omega \in \Omega$ is a realization of state. The \emph{privacy} is formalized as a random variable $\tilde{\theta}: \Omega \to \Theta$, where $(\Theta, \mathcal{B}(\Theta))$ is a standard Borel space. $\theta \in \Theta$ is a realization of privacy. For example, $\omega$ is the vector containing consumer characteristics such as gender, race, willingness to pay, and history records, and $\theta$ is the sensitive subvector consisting of gender and race.

A \emph{signal} is a random variable $\pi: \Omega \times [0,1] \to S$, where $(S,\mathcal{B}(S))$ is a standard Borel space. An element $s \in S$ is a \textit{signal realization}. Let $\tilde{r}$ be an auxiliary random variable uniformly distributed on $[0,1]$.
For each realization $(\omega, r)$, the signal realization is given by $\pi(\omega, r)$. Let $\lambda$ be the Lebesgue measure and $\mathbb{P}:=\mu_0 \times \lambda$ be the product measure induced by $\mu_0$ over $\Omega$ and $\lambda$ over $[0,1]$. Then the distribution of $\pi$ is $p^{\pi}(\cdot) := \mathbb{P}(\{(\omega, r):\pi(\omega,r)\in \cdot\})$ (i.e., $p^{\pi}(B) = \mathbb{P}(\{(\omega,r):\pi(\omega,r)\in B\})$, $\forall B \in \mathcal{B}(S)$, and similarly hereinafter).\footnote{$p^{\pi}$ is a probability measure on $(S, \mathcal{B}(S))$; it is called \textit{distribution} of $\pi$. If $S \subseteq \mathbb{R}$, then $F^{\pi}:\mathbb{R} \to [0,1]$, s.t. $F^{\pi}(s) = p^{\pi}(\pi \leq s)$ for all $s \in \mathbb{R}$ is called the \textit{distribution function} or CDF of $\pi$.} 
The conditional probability of $\pi$ given $\omega$ is
$p_{\omega}^{\pi}(\cdot):=\mathbb{P}(\pi\in \cdot|\omega)$.\footnote{Since $(\Omega \times [0,1], \mathcal{B}(\Omega) \otimes \mathcal{B}([0,1]))$ and $(S, \mathcal{B}(S))$ are standard Borel spaces, there exist regular versions of $\mathbb{P}(\pi\in\cdot|\omega)$ and $\mathbb{P}(\tilde{\omega}\in\cdot|s)$ respectively (\cite{ccinlar2011probability}, Theorem 2.19, p.154).}

Observing signal realization $s$ induces the posterior $\mu_s(\cdot):=\mathbb{P}(\tilde{\omega} \in \cdot|s)$.
Given signal $\pi$, let $\tilde{\mu}_{\pi}$ denote the associated belief-valued random variable of posterior belief $\mu_{s}$. Moreover, let $\left<\pi\right>(\cdot) := p^\pi(\{s \in S: \mu_s \in \cdot\})$ denote the distribution of $\tilde{\mu}_\pi$. We write $(\cdot)^{\theta}$ for the operator mapping a distribution over $\Omega$ or $\Delta(\Omega)$ to its marginal over $\Theta$ or $\Delta(\Theta)$. For $\mu\in\Delta(\Omega)$,
$\mu^{\theta}(\cdot) = \mu\bigl(\{\omega : \tilde{\theta}(\omega)\in\cdot\}\bigr)$,
and for $\tau\in\Delta(\Delta(\Omega))$, $\tau^{\theta}(\cdot) = \tau\bigl(\{\mu : \mu^{\theta}\in\cdot\}\bigr)$. Therefore, $\mu_s^{\theta}$ represents the posterior about privacy induced by signal realization $s$, and $\langle \pi\rangle^{\theta}$ represents the distribution of posteriors about privacy.

To keep the notation clear, we use $\mu\in\Delta(\Omega)$ to denote an arbitrary posterior over the state $\tilde{\omega}$, and $\nu\in\Delta(\Theta)$ to denote an arbitrary posterior over the privacy variable $\tilde{\theta}$. Likewise, we use $\tau\in\Delta(\Delta(\Omega))$ for an arbitrary distribution over posteriors about $\tilde{\omega}$, and $\gamma\in\Delta(\Delta(\Theta))$ for an arbitrary distribution over posteriors about $\tilde{\theta}$.

Denote $\Pi$ by the set of all signals. For two signals $\pi, \pi' \in \Pi$, we say that $\pi$ \emph{Blackwell dominates} $\pi'$ and write $\pi \succeq \pi'$ provided $\langle \pi \rangle$ is a \emph{mean-preserving spread} of $\langle \pi' \rangle$ which is also denoted as $\langle \pi \rangle \succeq \langle \pi' \rangle$ \citep{blackwell1951comparison,blackwell1953equivalent,strassen1965existence}.\footnote{\label{footnote:dilation} For $\tau,  \tau' \in \Delta(\Delta(\Omega))$, $\tau$ is a mean-preserving spread of $\tau'$ if there is a \emph{dilation} $K: \Delta (\Omega) \to \Delta(\Delta (\Omega))$ 
[i.e., a Markov kernel such that $\forall \mu' \in \operatorname{supp}(\tau')$,  $\mu' = \int_{\Delta (\Omega)} \mu K(d\mu|\mu')$ (mean-preservation)],
such that $\tau (\cdot) = \int_{\Delta (\Omega)} K(\cdot|\mu') \tau'(d\mu')$ (spread). The similar definition applies to any $\gamma, \gamma' \in \Delta(\Delta(\Theta))$.}
Given $\pi \succeq \pi'$, if $\pi \preceq \pi'$ also holds, then $\pi$ and $\pi'$ are \emph{Blackwell equivalent}, written $\pi \sim \pi'$; otherwise, $\pi$ \emph{strictly} Blackwell dominates $\pi'$, written $\pi \succ \pi'$. We say that $\pi$ Blackwell dominates $\pi'$ \emph{in terms of $\tilde{\theta}$} and write $\pi \succeq_{\theta} \pi'$ if $\langle \pi \rangle^{\theta}$ is a mean-preserving spread of $\langle \pi' \rangle^{\theta}$. The corresponding relations $\sim_{\theta}$ and $\succ_{\theta}$ are defined analogously. Whenever we refer to Blackwell dominance without the qualifier ``in terms of $\tilde{\theta}$,'' we mean dominance with respect to $\tilde{\omega}$.

\subsection{Privacy-Constrained Signals}

From an informational perspective, each signal can be identified with the posterior distribution it induces. Thus, analyzing signals is equivalent to working directly with Bayesian-plausible distributions over posteriors (cf. \cite{blackwell1953equivalent,kamenica2011bayesian}).

Let
$\Gamma := \{\gamma \in \Delta(\Delta(\Theta)) : \mathbb{E}_{\gamma}[\nu] = \mu_0^{\theta}\}$
denote the Bayesian-plausible distributions of posteriors about privacy.
If there is no constraint on how much information about privacy may be disclosed, then any posterior distribution in $\Gamma$ can be induced. Hence, in general, a privacy constraint can be imposed directly on the set $\Gamma$. 

In particular, let $$\mathcal{P} \subseteq \Gamma$$ be a nonempty subset containing the posterior
distributions about $\tilde{\theta}$ that are allowed to be disclosed. We refer to $\mathcal{P}$ as the \emph{privacy-permissible set}. To ensure that the privacy
constraint is well behaved, we impose the following assumptions on $\mathcal{P}$.

\begin{ass}\label{ass:blackwell_closeness}
$\mathcal{P}$ is a \emph{lower set} with respect to the Blackwell order; that is, if
$\gamma \in \mathcal{P}$ and $\gamma' \preceq \gamma$, then $\gamma' \in \mathcal{P}$.
\end{ass}

Assumption \ref{ass:blackwell_closeness} is natural and is required by any reasonable notion of a privacy constraint. 
It states that if a certain amount of information about privacy may be disclosed, then revealing any less informative disclosure must also be permissible.

\begin{ass}\label{ass:closed_set}
$\mathcal{P}$ is a \emph{closed set}; that is, if $\{\gamma_t\}_{t\in\mathbb{N}_+}$ satisfies
$\gamma_t \in \mathcal{P}$ for all $t \in \mathbb{N}_+$ and $\gamma^* := \lim_{t\to\infty}\gamma_t$ exists, then
$\gamma^* \in \mathcal{P}$.
\end{ass}

Assumption \ref{ass:closed_set} is a technical requirement that simplifies our characterization of privacy-constrained signals. When a decision-maker selects a signal in $\mathcal{P}$ to maximize an objective function, taking the supremum is equivalent to optimizing over the closure of $\mathcal{P}$. Let $$\overline{\mathcal{P}} := \{\gamma \in \mathcal{P}: \nexists \gamma' \in \mathcal{P} \text{ such that } \gamma' \succ \gamma\}$$ be the \emph{Blackwell frontier} of $\mathcal{P}$. Under Assumption~\ref{ass:blackwell_closeness} and \ref{ass:closed_set}, $\overline{\mathcal{P}} \neq \emptyset$ and $\mathcal{P} = \{\gamma \in \Gamma: \exists \overline{\gamma} \in \overline{\mathcal{P}} \text{ such that } \gamma \preceq \overline{\gamma}\}$.\footnote{Let $\{\gamma_t\}_{t \in \mathbb{N}_+}$ be a sequence in $\mathcal{P}$ such that $\gamma_t \prec \gamma_{t+1}$ for all $t$.
For each $t$, let $\tilde{\nu}_t$ be a belief-valued random variable satisfying $\tilde{\nu}_t \sim \gamma_t$.
Then $\{\tilde{\nu}_t\}_{t \in \mathbb{N}_+}$ is a martingale (rigorously, it requires $\gamma _{t+1}$ is sufficient for $\gamma _t$, see footnote~\ref{footnote:sufficiency}). By the martingale convergence theorem \citep{doob1951continuous}, the limit $\nu^* = \lim_{t\to\infty} \nu_t$ exists. Let $\gamma^*$ be the distribution of $\nu^*$. Then, $\lim_{t\to \infty} \gamma_t = \gamma^*$. By Assumption~\ref{ass:closed_set}, $\gamma^* \in \mathcal{P}$. Then, $\gamma^* \in \overline{\mathcal{P}}$. } $\overline{\mathcal{P}}$ gives several upper bounds of the amount of information can be disclosed about privacy. Moreover, when  $\mathcal{P}$ is not closed, our characterization of privacy-constrained signals can be made correct by removing some signals whose induced distribution over posterior about privacy is in the Blackwell frontier $\overline{\mathcal{P}}$.

\begin{definition}
    A signal $\pi$ is a $\mathcal{P}$-privacy-constrained signal if $\langle \pi \rangle^\theta \in \mathcal{P}$.
\end{definition}

\begin{example}[Privacy-Preserving Signals]
    When $\mathcal{P} = \{\delta_{\mu_0^{\theta}}\}$,\footnotetext{$\delta_{a}$ is the Dirac delta function, i.e., $\delta_a(x)=0$ if $x\neq a$ and $\int \delta_a(x) dx =1$.} $\mathcal{P}$-privacy-constrained signals reduce to the \textit{privacy-preserving signals} introduced by \citet{he2021private} and \citet{strack2024privacy}. \citeauthor{strack2024privacy} show that all privacy-preserving signals can be generated by garbling and reordering of a conditionally revealing quantile signal.
\end{example}

\begin{example}[Single-Bound Privacy]\label{eg:single_bound}
    A natural generalization of privacy-preserving signals is to consider $\mathcal{P} = \{\gamma \in \Gamma : \gamma \preceq \overline{\gamma}\}$ for some $\overline{\gamma} \in \Gamma$. The element $\overline{\gamma}$ provides a single upper bound on the amount of information that may be disclosed about privacy. A signal $\pi$ is said to be $\overline{\gamma}$-privacy-constrained if $\langle \pi\rangle^{\theta} \preceq \overline{\gamma}$.

\end{example}

This framework also includes differential privacy \citep{dwork2006calibrating}, inferential privacy \citep{ghosh2016inferential, wang2025inferentially}, and privacy constraints defined through the posterior mean of a statistic. These concepts are discussed in greater detail in Section~\ref{sec:privacy_permissible_set}.

\section{Characterization of Privacy-Constrained Signals}\label{sec:pcs}

Compare with privacy-preserving constraint, our $\mathcal{P}$-privacy constraint allows us to disclose some information about privacy. Based on this, we can intuitively construct the following two-stage disclosure rule:

\emph{Stage 1.} Construct a signal $\pi_{1}: \Omega \times [0,1] \to S_{1}$ to release permissible information about privacy $\tilde{\theta}$.

\emph{Stage 2.} Construct another signal $\pi_{2}: \Omega \times [0,1] \to S_{2}$, which is conditionally privacy-preserving given $\pi_1$, i.e., for almost every $B \in \mathcal{B}(\Theta)$ and $s_1 \in S_1, s_2 \in S_{2}$,
$$\mathbb{P}(\tilde{\theta} \in B | s_2, s_{1}) := \mu^{\theta}_{(s_{1}, s_{2})}(B)
= \mu^{\theta}_{s_{1}}(B).$$

Let $\pi_1 \lor \pi_2 : \Omega \times [0,1] \to S_1 \times S_2$ denote the \emph{join} of two signals $\pi_1$ and $\pi_2$. When the state–randomizer pair $(\omega,r)$ is realized, the joint signal is
$\pi_1 \lor \pi_2 (\omega,r) = (\pi_1(\omega,r), \pi_2(\omega,r))$. Observing $\pi_1 \lor \pi_2$ is equivalent to observing $\pi_1$ and then $\pi_2$ sequentially, with beliefs updated at each stage according to Bayes' rule. 
It is immediate that if $\pi_1$ is $\mathcal{P}$-privacy-preserving and $\pi_2$ is conditionally privacy-preserving given $\pi_1$, then the joint signal $\pi_1 \lor \pi_2$ is $\mathcal{P}$-privacy-preserving as well. Conditional privacy preservation corresponds to designing $\pi_2$ separately for each realization of $\pi_1$; see Remark 2 in \cite{strack2024privacy}. In other words, constructing $\pi_2$ amounts to constructing a privacy-preserving signal on each posterior belief induced by $\pi_1$.

Denote by $\Pi_{\mathcal{P}}$ the set of $\mathcal{P}$-privacy-constrained signals, and define its Blackwell frontier as $\overline{\Pi}_{\mathcal{P}} :=
\left\{
\pi \in \Pi_{\mathcal{P}}
: \nexists \pi' \in \Pi_{\mathcal{P}}
\text{ such that }
\pi' \succ \pi
\right\}$.
Any $\pi \in \overline{\Pi}_{\mathcal{P}}$ is called a \emph{Blackwell-undominated} $\mathcal{P}$-privacy-constrained signal. In what follows, we present our main result, Theorem~\ref{thm:decomposition}, which characterizes the Blackwell frontier $\overline{\Pi}_{\mathcal{P}}$ via two-stage disclosure rules. Any $\mathcal{P}$-privacy-constrained signal can then be obtained by garbling elements of this frontier.

\begin{lemma}\label{lemma:blackwell-frontier-P}
     A signal is $\mathcal{P}$-privacy-constrained if and only if it is Blackwell dominated by a signal $\pi'$ such that $\langle \pi' \rangle^\theta \in \overline{\mathcal{P}}$.
\end{lemma}

Lemma \ref{lemma:blackwell-frontier-P} states that, without loss of generality, constructing the Blackwell frontier of $\mathcal{P}$-privacy-constrained signals reduces to considering only those signals whose induced distribution over posteriors about privacy lies on the Blackwell frontier of $\mathcal{P}$. Then, by garbling these  signals in the frontier, we obtain all $\mathcal{P}$-privacy-constrained signals. Therefore, in what follows, we focus on constructing the Blackwell frontier of $\mathcal{P}$-privacy-constrained signals from the Blackwell frontier of $\mathcal{P}$.

For each $\gamma \in \overline{\mathcal{P}}$, we construct its \emph{minimum-informative extensions}, namely every posterior about state $\tau_{\gamma} \in \mathcal{T}$ satisfying
\begin{equation*}
    \begin{aligned}
        \tau_\gamma^\theta = \gamma \text{ almost surely} \quad & \text{ (Extension)}, \\
        \nexists \tau \in \mathcal{T} \text{ such that } \tau^\theta = \gamma \text{ almost surely},  \text{ and } \tau \prec \tau_\gamma \quad & \text{  (Minimum informative)}.
    \end{aligned}
\end{equation*}
That is, $\tau_\gamma$ extends $\gamma$ from a distribution over posteriors about privacy $\tilde{\theta}$ to a distribution over posteriors about the full state $\tilde{\omega}$, while revealing the least additional information about the state needed to sustain the same distribution of posteriors about privacy. Since a posterior distribution can itself be viewed as a signal, the object $\tau_\gamma$ provides exactly the first–stage disclosure we seek.
To simplify notation, we use $\tau_\gamma$ to denote any signal $\pi$ such that $\langle \pi \rangle = \tau_\gamma$.\footnote{Formally, $\tau$ can induce a joint distribution over $\Omega\times \Delta(\Omega)$, $p^{(\tilde{\omega},\tilde{\mu})}$ such that for any $B_{\Omega} \times B_{\Delta(\Omega)} \in \mathcal{B}(\Omega)\otimes \mathcal{B}(\Delta(\Omega))$, $p^{(\tilde{\omega},\tilde{\mu})}(B_{\Omega} \times B_{\Delta(\Omega)}) = \int_{B_{\Delta(\Omega)}} \mu(B_{\Omega}) \tau(d\mu)$. Since $\Delta(\Omega)$ with the Borel $\sigma$-algebra induced by the weak-$*$ topology, is standard Borel, there exists a regular version of conditional distribution of $\tilde{\mu}$ given $\omega$, denoted by $p_{\omega}^{\tilde{\mu}}$ (\cite{ccinlar2011probability}, Theorem 2.18, p.154). Then there exists a measurable function $\pi_\tau:\Omega \times [0,1] \to \Delta(\Omega)$, such that $\pi_\tau(\omega,\tilde{r})$  has distribution $p_\omega^{\tilde{\mu}}$ (\cite{kallenberg1997foundations}, Lemma 2.22, p.34). Therefore, $\pi_\tau$ is the signal with $\langle\pi_\tau\rangle = \tau$.} For a $\gamma \in \overline{\mathcal{P}}$, the minimum-informative extension need not be unique. Let $\underline{\mathcal{T}}_\gamma$ denote the set of all minimum-informative extensions of $\gamma$.

\begin{theorem}[Characterization of Privacy-Constrained Signals]\label{thm:decomposition}
A signal $\pi$ belongs to the Blackwell frontier of $P$-privacy-constrained signals, $\overline{\Pi}_\mathcal{P}$, if and only if there exist $\tau_{\gamma} \in \underline{\mathcal{T}}_{\gamma}$ for some $\gamma \in \overline{\mathcal{P}}$ and  $\tilde{q}_{\tau_{\gamma}}$ that is Blackwell-undominated conditionally privacy-preserving given $\tau_\gamma$ such that $\pi$ is Blackwell equivalent to the joint signal $\tau_{\gamma} \lor \tilde{q}_{\tau_{\gamma}}$.
\end{theorem}

Briefly, Theorem \ref{thm:decomposition} shows that, once $\tau_{\gamma}$ is made public, the remaining task of constructing the Blackwell frontier of $\mathcal{P}$-privacy-constrained signals reduces to constructing the Blackwell frontier of privacy-preserving signals, as characterized in \citet{strack2024privacy}. 

The key insight of our approach is that, in order to construct the most informative signal whose distribution over posteriors about privacy is $\gamma$, we begin by identifying the least informative one. Once this baseline signal is obtained, every signal with the same $\gamma$ can be generated by disclosing additional conditionally privacy-preserving information. By contrast, if one starts with a signal that already contains unnecessary information beyond what is required to sustain $\gamma$, then taking its join with any conditionally privacy-preserving signal inevitably preserves this extra information, and thus cannot characterize the full set of feasible signals sustaining $\gamma$. Figure~\ref{fig:pcs} summarizes the structure of Blackwell-undominated privacy-constrained signals described in Theorem~\ref{thm:decomposition}.

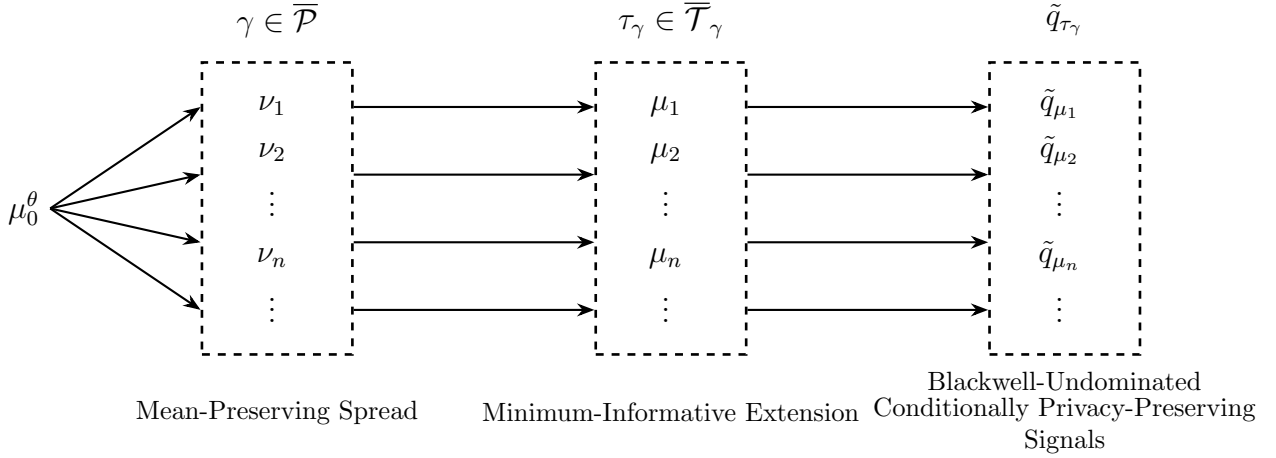
\begin{figure}[t]
\centering
\begin{tikzpicture}[
    >=Stealth,
    font=\normalsize,
    box/.style={
        draw, dashed,
        line width=1pt,      
        minimum width=2.0cm, 
        minimum height=3.6cm,
        inner sep=8pt,
        align=center
    },
    lab/.style={font=\footnotesize, align=center},
    arr/.style={->, thick}
]

\node (mu0) at (0,0) {$\mu_0^\theta$};

\node[box, right=2.0cm of mu0] (gamma) {
    \begin{tabular}{c}
        $\nu_1$\\[3pt]
        $\nu_2$\\[3pt]
        $\vdots$\\[3pt]
        $\nu_n$\\[3pt]
        $\vdots$
    \end{tabular}
};
\node at ($(gamma.north)+(0,0.55)$) {$\gamma\in\overline{\mathcal{P}}$};

\foreach \y in {1.35,0.45,-0.45,-1.35}{
  \draw[arr] (mu0.east) -- ($(gamma.west)+(0,\y)$);
}

\node[box, right=3.2cm of gamma] (tau) {
    \begin{tabular}{c}
        $\mu_1$\\[3pt]
        $\mu_2$\\[3pt]
        $\vdots$\\[3pt]
        $\mu_n$\\[3pt]
        $\vdots$
    \end{tabular}
};
\node at ($(tau.north)+(0,0.55)$) {$\tau_\gamma\in\overline{\mathcal{T}}_\gamma$};

\foreach \y in {1.35,0.45,-0.45,-1.35}{
  \draw[arr] ($(gamma.east)+(0,\y)$) -- ($(tau.west)+(0,\y)$);
}

\node[box, right=3.2cm of tau] (q) {
    \begin{tabular}{c}
        $\tilde q_{\mu_1}$\\[3pt]
        $\tilde q_{\mu_2}$\\[3pt]
        $\vdots$\\[3pt]
        $\tilde q_{\mu_n}$\\[3pt]
        $\vdots$
    \end{tabular}
};
\node at ($(q.north)+(0,0.55)$) {$\tilde q_{\tau_\gamma}$};

\foreach \y in {1.35,0.45,-0.45,-1.35}{
  \draw[arr] ($(tau.east)+(0,\y)$) -- ($(q.west)+(0,\y)$);
}

\node[lab] at ($(gamma.south)+(0,-0.75)$) {Mean-Preserving Spread};
\node[lab] at ($(tau.south)+(0,-0.75)$) {Minimum-Informative Extension};
\node[lab] at ($(q.south)+(0,-0.75)$) {\shortstack{Blackwell-Undominated \\ Conditionally Privacy-Preserving\\Signals}};

\end{tikzpicture}
\caption{Construction of Blackwell-Undominated Privacy-Constrained Signals}
\label{fig:pcs}
\end{figure}

Additionally, the characterization of the minimum-informative extension of a given $\gamma \in \overline{\mathcal{P}}$ is straightforward.

\begin{theorem}[Characterization of Minimum-Informative Extensions]
\label{prp:least-informative-signal}
    $\underline{\tau} \in \mathcal{T}$ is a minimum-informative extension of a $\gamma \in \overline{\mathcal{P}}$ if and only if 
    \begin{enumerate}[(1)]
        \item\label{least.1} $\underline{\tau}^\theta = \gamma$ almost surely;
        
        \item\label{least.2} For almost every $\mu, \hat{\mu} \in \operatorname{supp}(\underline{\tau})$, if $\mu^\theta = \hat{\mu}^\theta$ almost surely, then $\mu = \hat\mu$ almost surely.
    \end{enumerate}
\end{theorem}

The key condition in Theorem~\ref{prp:least-informative-signal} is~\eqref{least.2}, which requires that, in the minimum-informative extension, each $\nu \in \operatorname{supp}(\gamma)$ extends to exactly one $\mu \in \Delta(\Omega)$. Otherwise, any two such $\mu$'s with the same marginal $\mu^{\theta}$ can be merged into a single posterior, implying that the original extension was not minimal. 

Theorem \ref{prp:least-informative-signal} is particularly useful in the discrete setting. Consider finite sets $\Omega = \{\omega_i\}_{i=1}^I$,  $\Theta = \{\theta_j\}_{j=1}^J$ with $J < I$ and let $\operatorname{supp}(\gamma) = \{\nu_n\}_{n=1}^N$. By Theorem \ref{prp:least-informative-signal}, characterizing the minimum-informative extensions of $\gamma$ amounts to allocating, for each $\nu \in \operatorname{supp}(\gamma)$, a conditional distribution over $\Omega$ given $\Theta$. For any $\mu \in \Delta(\Omega)$, denote by $\mu(\omega_i|\theta_j)$ the conditional probability of $\omega_i$ given $\theta_j$. Any sequence of posterior distributions $\{\mu_n\}_{n=1}^N$ constitutes a minimum-informative extension of $\gamma$ if and only if it satisfies the following conditions:
\begin{gather}
    \mu_n(\omega_i) = \sum_{j=1}^J\nu_n(\theta_j) \mu_n(\omega_i|\theta_j), \quad \text{for all } n, \label{eq:min_ext_1}\footnotemark
    \\
    \mu_0(\omega_i|\theta_j) = \sum_{n=1}^N \frac{\gamma(\nu_n)\nu_n(\theta_j) \mu_n(\omega_i|\theta_j)}{\sum_{n=1}^N \gamma(\nu_n) \nu_n(\theta_j)}, \quad \text{for all } i,j,\label{eq:min_ext_2}\\
    \sum_{i=1}^I \mu_n(\omega_i|\theta_j) = 1, \quad \text{for all } j,n, \label{eq:min_ext_3}\\
    \mu_n(\omega_i|\theta_j) \geq 0, \quad \text{for all } i,j,n.\label{eq:min_ext_4}
\end{gather}
\footnotetext{Since $\tilde{\theta}$ is a measurable function of $\Omega$, the realization of $\omega$ uniquely determines $\theta$. Hence $\mu_n(\omega_i)=\mu_n(\omega_i,\theta_j)$, and $\mu_n(\omega_i,\theta_j)>0$ only if $\theta_j=\tilde{\theta}(\omega_i)$. Thus, in Equation~\eqref{eq:min_ext_1}, at most one term in the summation can be positive.}
Given any collection $\{\mu_n\}_{n=1}^N$ satisfying these constraints, the associated minimum-informative extension is the distribution $\tau_\gamma$ over $\{\mu_1\}_{n=1}^N$ defined by $\tau_\gamma(\mu_n) = \gamma(\mu_n^\theta) = \gamma(\nu_n)$ for all $n$. Equation~\eqref{eq:min_ext_1} ensures that, for each $\nu_n$, there exists a unique posterior $\mu_n$ such that $\mu_n^{\theta}=\nu_n$. By the construction of $\tau_\gamma$, it follows immediately that $\tau_\gamma^{\theta}=\gamma$. Equation~\eqref{eq:min_ext_2} guarantees that the induced distribution $\tau_\gamma$ satisfies $\mathbb{E}_{\tau_\gamma}[\mu]=\mu_0$. Finally, Equations~\eqref{eq:min_ext_3} and \eqref{eq:min_ext_4} ensure that $\mu_n(\omega_i |\theta_j)$ is a well-defined probability distribution.

To illustrate how our characterization can be used to describe all privacy‐constrained signals, we present a simple example.

\begin{example}\label{ex:2x2}
    Consider the following $2 \times 2$ setting. The state space $\Omega = X \times \Theta$ consists of a variable of interest $X = \{x_1, x_2\}$ and a privacy  $\Theta = \{\theta_1, \theta_2\}$. Define $\overline{\gamma} \in \Gamma \setminus\{\delta_{\mu_0^\theta}\}$ by $\overline{\gamma} = \alpha \delta_{\nu_1} + (1-\alpha) \delta_{\nu_2}$. Now we construct all Blackwell-undominated $\overline{\gamma}$-privacy-constrained signals (Example~\ref{eg:single_bound}).
    
    First, by Theorem~\ref{prp:least-informative-signal}, any minimum-informative extension $\tau_\gamma$ of $\overline{\gamma}$ takes the form $\tau_\gamma = \alpha \delta_{\mu_1} + (1-\alpha)\delta_{\mu_2}$, where $\mu_1$ and $\mu_2$ satisfy
    \begin{gather*}
        \mu_1((x_i,\theta_j)) = \nu_1(\theta_j)\mu_1(x_i|\theta_j), \quad
        \mu_2((x_i,\theta_j))  = \nu_2(\theta_j) \mu_2(x_i|\theta_j), \quad \text{for all } i,j.  
    \end{gather*}
    with the pairs $\{\mu_n(x_i|\theta_j)\}_{i,j,n=1}^2$ that
    \begin{gather*}
        \mu_2(x_1|\theta_j) = \mu_0(x_1|\theta_j) + \frac{\alpha \nu_1(\theta_j)}{(1-\alpha)\nu_2(\theta_j)}[\mu_0(x_1|\theta_j) - \mu_1(x_1|\theta_j)], \quad \text{for all } j, \\ 
        \mu_n({x_1|\theta_j}) + \mu_n(x_2|\theta_j) = 1, \quad
        \mu_n(x_i|\theta_j) \geq 0, \quad \text{for all } i,j,n.
    \end{gather*}
    
    Second, applying Theorem~\ref{thm:decomposition}, we construct the Blackwell-undominated conditionally privacy-preserving signals associated with each $\tau_\gamma$. For each realization $\mu_n \in \tau_\gamma$, Theorem~1 in \cite{strack2024privacy} implies that there exists a quantile signal $\tilde{q}^*_{\mu_n}$ such that 
    $$\tilde{q}^*_{\mu_n}(x_1,\theta_j) \sim U\left[0, \mu_n(x_1 | \theta_j)\right], \quad
    \tilde{q}^*_{\mu_n}(x_2,\theta_j) \sim U\left[\mu_n(x_1 | \theta_j), 1\right],$$ which is Blackwell-undominated since revealing the realization of $\tilde{\theta}$ would allow the receiver to infer the state $\tilde{\omega}$ exactly. By reordering the quantile signal, we can generate all signals that are Blackwell-undominated and conditionally privacy-preserving given $\mu_n$. Specifically, any such reordered quantile signal, denoted by $\tilde{q}_{\mu_n}$, satisfies $$\tilde{q}_{\mu_n}(x_1,\theta_j) \sim U(I_{\mu_n(x_1\mid\theta_j)}),  \quad
    \tilde{q}_{\mu_n}(x_2,\theta_j) \sim U\left([0,1]\setminus I_{\mu_n(x_1\mid\theta_j)}\right),$$
    where $I_a \subseteq [0,1]$ is any measurable subset with Lebesgue measure $\lambda(I_a)=a$.

    Let $\tilde{q}_{\tau_\gamma}$ denote the collection of these reordered quantile signals across realizations of $\tau_\gamma$. The set of all $\tau_\gamma \lor \tilde{q}_{\tau_\gamma}$ then constitutes the Blackwell frontier of $\overline{\gamma}$–privacy-constrained signals.
\end{example}

Now consider a decision-making problem. For a decision maker who seeks to maximize an objective subject to the $\gamma$-privacy constraint, the resulting optimization problem can be solved in two steps. First, under each minimum-informative extension, one solves an optimal transport problem \citep{strack2024privacy}. Second, one optimizes over the set of minimum-informative extensions, which reduces to a linear programming problem. For a general $\mathcal{P}$-privacy constraint, however, an additional third step is required: one must also optimize over the Blackwell frontier of $\mathcal{P}$.

\section{Blackwell Frontier of Privacy-Permissible Set}\label{sec:privacy_permissible_set}

For a general privacy-permissible set $\mathcal{P}$, characterizing the $\mathcal{P}$-privacy-constrained signals first requires characterizing the Blackwell frontier of $\mathcal{P}$. In some cases, the Blackwell frontier is explicitly given, as in Example~\ref{eg:single_bound}. In other cases, such as differential privacy, the set $\mathcal{P}$ is defined by a collection of constraints, and its Blackwell frontier is not directly evident.  Unfortunately, there is no uniform method for characterizing the Blackwell frontier of an arbitrary or abstract $\mathcal{P}$. In this section, we focus on two important classes for which the Blackwell frontier can be analyzed: ex-post privacy and posterior-mean privacy.

\subsection{Ex-Post Privacy}

When the regulator cares about the \emph{ex post} cost of disclosing information about privacy, the constraint applies not to the distribution over posteriors about privacy, but directly to the posteriors themselves. Formally, let
$$\mathcal{M} \subseteq \Delta(\Theta)$$
denote the set of permissible posteriors about privacy. In line with the spirit of Blackwell-closeness in Assumption~\ref{ass:blackwell_closeness}, we assume that $\mathcal{M}$ is a compact convex subset of $\Delta(\Theta)$. If a signal $\pi$ is permissible, meaning that every posterior $\nu \in \operatorname{supp}(\langle \pi \rangle^{\theta})$ lies in $\mathcal{M}$, then any less informative signal $\pi' \preceq \pi$ is also permissible. It follows that every convex combination of posteriors in $\operatorname{supp}(\langle \pi \rangle^{\theta})$ must also lie in $\mathcal{M}$. To avoid triviality, we assume that $\mu_{0}^{\theta} \in \mathcal{M}$.

\begin{definition}
    A signal $\pi$ is a $\mathcal{M}$-ex-post-privacy-constrained signal if $\langle \pi \rangle^\theta (\mathcal{M}) = 1$.
\end{definition}

Given $\mathcal{M}$, the induced privacy-permissible set is
$\mathcal{P}_{\mathcal{M}}:=\{\gamma \in \mathcal{T} : \gamma(\mathcal{M}) = 1\}$. Since $\mu_{0}^{\theta} \in \mathcal{M}$, the set $\mathcal{P}_{\mathcal{M}}$ is nonempty. Because $\mathcal{M}$ is compact and convex, the set $\mathcal{P}_{\mathcal{M}}$ satisfies Assumptions~\ref{ass:blackwell_closeness} and~\ref{ass:closed_set}. Therefore, $\mathcal{M}$-ex-post privacy is a special case of $\mathcal{P}$-privacy. 

Let $\operatorname{ext} \mathcal{M}$ as the the set of extreme points of $\mathcal{M}$, i.e, $$\operatorname{ext} \mathcal{M}:=\{\nu \in \mathcal{M}: \nexists \nu' \neq \nu'' \in \mathcal{M}, \alpha \in (0,1) \text{ such that } \nu = \alpha \nu' + (1-\alpha) \nu'' \}.$$
Let $\overline{\mathcal{P}}_{\mathcal{M}}$ denote the Blackwell frontier of $\mathcal{P}_{\mathcal{M}}$ as defined previously.

\begin{proposition}[Characterization of Blackwell Frontier of Ex-Post Privacy-Permissible Set]\label{prp:extreme_points}
    A distribution of posteriors about privacy $\gamma \in \Gamma$ belongs to  $\overline{\mathcal{P}}_{\mathcal{M}}$ if and only if $\gamma(\operatorname{ext} \mathcal{M}) = 1$.
\end{proposition}

\begin{remark}
    When $\operatorname{ext}\mathcal{M}$ is finite and we restrict attention to distributions $\gamma$ with finite support, Proposition~\ref{prp:extreme_points} becomes immediate. The general case requires substantially more work. For the ``only if'' direction, one must construct a dilation (see footnote~\ref{footnote:dilation}) to show that any distribution $\gamma'$ with $\gamma'(\operatorname{ext}\mathcal{M})<1$ can be mean-preserving spread to a distribution $\gamma$ satisfying $\gamma(\operatorname{ext}\mathcal{M})=1$.
    This requires two ingredients. First, for each $\nu \in \mathcal{M}$ there must exist a probability distribution $P_\nu$ over $\operatorname{ext}\mathcal{M}$ such that $\mathbb{E}_{P_\nu}[\nu']=\nu$, which follows from Choquet's Theorem (Theorem 10.7, p.168, \cite{simon2011convexity}). Second, the mapping $\nu \mapsto P_\nu$ from $\mathcal{M}$ to $\Delta(\operatorname{ext}\mathcal{M})$ must be measurable. This measurability is ensured by Theorem~9.1 in \citet{simon2011convexity} (p.136), together with the Kuratowski–Ryll-Nardzewski measurable selection theorem (Theorem~6.9.3, p.36, Vol.II, \cite{bogachev2007measure}).
\end{remark}

By Proposition~\ref{prp:extreme_points}, under $\mathcal{M}$-ex-post privacy the Blackwell frontier of permissible distributions over posteriors about privacy, $\overline{\mathcal{P}}_{\mathcal{M}}$, can be generated by first computing the extreme points of the permissible posterior set $\mathcal{M}$. Any distribution in $\overline{\mathcal{P}}_{\mathcal{M}}$ is then obtained by convex combinations of these extreme points with respect to the prior $\mu_0^\theta$.

Ex-post privacy encompasses all privacy notions that impose constraints directly on posterior beliefs about privacy and that satisfy Blackwell-closeness. When privacy realizations are finite and the constraints imposed on posteriors are finite and linear, the set $\mathcal{M}$ becomes a convex polytope whose extreme points are finite and can be computed explicitly. For example, the well-known concept of differential privacy introduced by \citet{dwork2006calibrating} is defined by a finite collection of linear constraints on posterior beliefs about privacy. This concept has been adopted by major institutions such as Google, Microsoft, and the U.S. Census Bureau \citep{abowd2018us}.  The characterization of Blackwell frontier of differential privacy has been conducted by \cite{xu2025privacy}.

Similar with differential privacy, \cite{ghosh2016inferential} introduces the notion of inferential privacy in the context of vector-valued datasets. \cite{wang2025inferentially} adapts this concept to the framework of \cite{he2021private}, which corresponds to the discrete case of our setting. Our Definition \ref{def:ip} generalizes inferential privacy to arbitrary standard Borel spaces.


\begin{definition}\label{def:ip}
    For any $\varepsilon \in [0,+\infty)$. A signal $\pi$ is $\varepsilon$-inferential-privacy-constrained if for almost every $\nu \in \operatorname{supp}(\langle \pi \rangle^\theta)$ and $B', B'' \in \mathcal{B}(\Theta)$ such that $\mu_0(B') > 0$, $\mu_0(B'') > 0$,
    \begin{equation}\label{eq:ip}
        \frac{\nu(B')}{\nu(B'')} \leq 
        e^\varepsilon \cdot\frac{\nu(B')}{\nu (B'')}.
    \end{equation} 
\end{definition}

Let $\mathcal{I}$ denote the set of posteriors that satisfy the $\varepsilon$-inferential privacy constraint (\ref{eq:ip}), $\mathcal{P}_{\mathcal{I}}$ the privacy-permissible set induced by $\mathcal{I}$, and $\overline{\mathcal{P}}_{\mathcal{I}}$ its Blackwell frontier.

\begin{proposition}\label{prp:ip-froniter}
    A distribution of posteriors about privacy $\gamma \in \Gamma$ belongs to $\overline{\mathcal{P}}_\mathcal{I}$ if and only if
    for almost every $\nu \in \operatorname{supp}(\gamma)$, there is a subset $E_\nu \in \mathcal{B}(\Theta)$ such that $\mu_0^{\theta} (E_\nu) \in (0,1)$, and 
    \begin{equation}\label{prp_Dichotomy}
        \nu(B) = \frac{e^\varepsilon \mu_0^{\theta}(B \cap E_\nu) + \mu_0^\theta(B \setminus E_\nu)}{e^\varepsilon \mu_0^{\theta}(E_\nu) + (1-\mu_0^{\theta}(E_\nu))} ,
    \end{equation}
    for almost every $B \in \mathcal{B}(\Theta)$.
\end{proposition}

Proposition~\ref{prp:ip-froniter} provides a clean characterization of the extreme points under $\varepsilon$-inferential privacy, a setting with infinitely many privacy realizations and linear constraints. An extreme point of $\mathcal{I}$ partitions the privacy realization space into two sets of positive measure, $E$ and $\Theta \setminus E$. Relative to the prior, the posterior probability assigned to $E$ is increased uniformly, while the posterior probability assigned to $\Theta \setminus E$ is decreased uniformly. Using Proposition~\ref{prp:ip-froniter}, we can simplify the main characterizations of $\varepsilon$-inferential-private private information structure presented in \citet{wang2025inferentially}.


\subsection{Posterior-Mean Privacy}

In some settings, an individual's privacy loss depends only on the posterior expected type rather than on the full posterior distribution. In other words, the privacy cost associated with each posterior belief is a linear function of its realization. In this case, a natural privacy constraint is to impose an upper bound on the distribution of posterior means of a statistic of the privacy.

Specially, let $\tilde{f}:\Theta \to \mathbb{R}$ be a one-dimensional statistic of privacy. Denote by $\nu^{f}_{0}$ the prior distribution of $\tilde{f}$, and suppose that $\bar{\kappa}$ is a mean-preserving contraction of $\nu^{f}_{0}$. A signal $\pi$ is $f$-posterior-mean-privacy-constrained (with respect to $\bar{\kappa}$) if the distribution over posterior means of $\tilde{f}$ is a mean-preserving contraction of $\bar{\kappa}$. For any $\gamma \in \Gamma$, denote the induced distribution of posterior mean about $\tilde{f}$ by $$\kappa_\gamma(\cdot):=\gamma(\{\nu \in \operatorname{supp}(\gamma):\mathbb{E}_\gamma [\tilde{f}(\theta)] \in \cdot\}).$$ When $\Theta \subseteq \mathbb{R}$, the statistic $\tilde{f}$ captures all moments of $\tilde{\theta}$. Since the distribution of posterior means of $\tilde{f}$ depends on the  distribution of posteriors about privacy, posterior-mean privacy is not an ex-post privacy.

\begin{definition}
    A signal $\pi$ is a $f$-posterior-mean-privacy-constrained signal if $\kappa_{\langle \pi \rangle^\theta} \preceq \overline{\kappa}$.
\end{definition}

Let $\mathcal{E}$ denote the set of posteriors that satisfy the $f$-posterior-mean privacy, $\mathcal{P}_{\mathcal{E}}$ the privacy-permissible set induced by $\mathcal{E}$, and $\overline{\mathcal{P}}_{\mathcal{E}}$ its Blackwell frontier.

\begin{proposition}\label{prp:expectation}
    A distribution of posteriors about privacy $\gamma \in \Gamma$ belongs to $\overline{\mathcal{P}}_{\mathcal{E}}$ if and only if (1) 
    for almost every $\nu \in \operatorname{supp} (\gamma)$, there exists $y_1, y_2 \in \tilde{f}(\Theta)$ and $\alpha \in (0,1]$, $\nu$ puts $\alpha$ on a point $\theta_1 \in \tilde{f}^{-1}(y_1)$ and $(1-\alpha)$ on another point $\theta_2 \in \tilde{f}^{-1}(y_2)$ and (2) $\kappa_\gamma = \overline{\kappa}$.
\end{proposition}

Proposition~\ref{prp:expectation} characterizes the Blackwell frontier of the privacy-permissible set induced by posterior-mean privacy. When $\gamma \in \overline{\mathcal{P}}_\mathcal{\mathcal{E}}$, every realized posterior about privacy is a two-point distribution, and the induced distribution over posterior mean about $\tilde{f}$ attains the upper bound.

\section{Discussion and Future Work}\label{sec:discussion}
We provide a characterization of signals that, in the Blackwell sense, do not reveal more private information than those in a given privacy-permissible set $\mathcal{P}$, where $\mathcal{P}$ is a subset of distributions over posterior beliefs about privacy. Specifically, we show that the most informative $\mathcal{P}$-privacy-constrained signals can be constructed as the join of two components: (i) a minimum-informative extension of a distribution on the Blackwell frontier of $\mathcal{P}$, and (ii) a Blackwell-undominated, conditionally privacy-preserving signal as characterized by \citet{strack2024privacy}. We then characterize the Blackwell frontier of the privacy-permissible set, with particular attention to ex-post privacy and posterior-mean privacy.

This paper has several limitations. First, we do not provide a general method for identifying the Blackwell frontier of an arbitrary privacy-permissible set. Our current results apply only to two specific formulations of privacy. Extending these results to more general notions of privacy remains an important direction for future research.

Second, for the general privacy constraint, when there is a unique bound on privacy information, our approach reduces the decision-making problem to a first-stage optimal transport problem followed by a second-stage linear program. Although this method is conceptually clean, it is computationally demanding. Developing a more tractable solution approach remains an important direction for future research.

\bibliographystyle{chicagoa}
\bibliography{pcs}

\appendix

\section{Appendix}
\subsection{Proofs for Section~\ref{sec:pcs}}\label{appendix:thm}

\begin{lemma}\label{lemma:marginal_dominate}
    For two signals $\pi, \pi'$, if $\pi \succeq \pi'$, then $\pi \succeq_{\theta} \pi'$.
\end{lemma}
\begin{proof}
    Since $\pi \succeq \pi'$, then there exists a dilation $K:\Delta(\Omega) \to \Delta(\Delta(\Omega))$, such that for almost every $\mu_{s'} \in \operatorname{supp}(\langle \pi'\rangle)$, $\mu_{s'} = \int_{\Delta(\Omega)} \mu_s dK(\mu_s|\mu_{s'})$, and $\langle \pi \rangle(\cdot) = \int_{\Delta(\Omega)} K(\cdot|\mu_{s'}) d\langle \pi'\rangle(\mu_{s'})$ (see footnote~\ref{footnote:dilation}). Define $Q:\Delta(\Theta) \to \Delta(\Delta(\Theta))$, such that for $\langle \pi'\rangle^\theta$-almost every $\nu' \in \Delta(\Theta)$ and all $B \in \mathcal{B}(\Delta(\Theta))$,
    \begin{equation*}
        Q(B|\nu') =  \mathbb{E}_{\langle \pi'
        \rangle}\left[K(\{\mu \in \Delta(\Omega): \mu^\theta \in B \} | \mu_{s'}) \mid \mu_{s'}^\theta = \nu'\right].
    \end{equation*}

    Then, for almost every $\nu' \in \operatorname{supp}(\langle \pi' \rangle ^ \theta)$,
    \begin{align*}
        \nu'  & = \mathbb{E}_{\langle \pi' \rangle} \left[ \mu_{s'}^\theta \mid \mu_{s'}^\theta = \nu' \right] \\
        & = \mathbb{E}_{\langle \pi' \rangle} \left[\int_{\Delta(\Omega)} \mu_s^\theta dK(\mu_s | \mu_{s'}) \mid \mu_{s'}^\theta = \nu' \right] 
        \\
        & = \mathbb{E}_{\langle \pi'
        \rangle}\left[ \int_{\Delta(\Theta)} \nu dK(\{\mu \in \Delta(\Omega): \mu^\theta = \nu \} | \mu_{s'}) \mid \mu_{s'}^\theta = \nu'\right]\\
        & = \int_{\Delta(\Theta)} \nu \mathbb{E}_{\langle \pi'
        \rangle}\left[ dK(\{\mu \in \Delta(\Omega): \mu^\theta = \nu \} | \mu_{s'}) \mid \mu_{s'}^\theta = \nu'\right] \\
        & = \int_{\Delta(\Theta)} \nu dQ(\nu|\nu').
    \end{align*}
    For almost every $B \in \mathcal{B}(\Delta(\Theta))$,
    \begin{align*}
        \langle \pi \rangle^{\theta}(B) &= \langle \pi \rangle(\{\mu \in \Delta(\Omega) : \mu^{\theta} \in B\}) \\
        & = \mathbb{E}_{\langle \pi' \rangle} \left[K(\{\mu \in \Delta(\Omega): \mu^{\theta} \in B\}|\mu_{s'})\right] \\
        & = \int_{\Delta(\Theta)} \mathbb{E}_{\langle \pi' \rangle}\left[K(\{\mu \in \Delta(\Omega): \mu^\theta \in B\} | \mu_{s'}) \mid \mu_{s'}^\theta = \nu'\right] d\langle \pi' \rangle^\theta (\nu') \\
        & = \int_{\Delta(\Theta)} Q(B|\nu') d\langle \pi' \rangle(\nu').
    \end{align*}
    Therefore, 
    $\pi \succeq_{\theta} \pi'$.
\end{proof}

For two signals $\pi$ and $\pi'$, we say that $\pi$ is \textit{sufficient} for $\pi'$ if $\pi \sim \pi \lor \pi'$. It is equal to say that given $\pi$, $\pi'$ is independent with $\tilde{\omega}$, i.e., $\mathbb{P}(\pi' \in \cdot|\omega, s) = \mathbb{P}(\pi' \in \cdot|s)$ almost surely.\footnote{By the definition, $\pi \sim \pi \lor \pi'$ says that $\mathbb{P}(\tilde{\omega} \in \cdot | s, s') = \mathbb{P}(\tilde{\omega} \in \cdot|s)$ almost surely. According to Theorem 7 in \cite{blackwell1951comparison}, it is equivalent to $\mathbb{P}(\pi' \in \cdot |\omega, s) = \mathbb{P}(\pi' \in \cdot |s)$ almost surely.} 
Furthermore, if $\pi \succeq \pi'$, there exists $\pi'' \sim \pi$ such that $\pi''$ is sufficient for $\pi'$ (\cite{blackwell1951comparison}, Theorem 6).\footnote{\label{footnote:sufficiency}The key distinction between Blackwell domination and sufficiency is that sufficiency requires the process from $\tilde{\mu}_{\pi'}$ to $\tilde{\mu}_{\pi}$ to form a martingale. Blackwell domination, on the other hand, does not impose this martingale requirement, allowing for more general transformations between signals. To understand it, consider the following example:

Under signal $\pi$, the posterior belief about $\tilde{\omega}$ takes the value $0$ (i.e, a distribution putting $1$ at $\omega=0$) with prob $1/4$, $1/2$ with prob $1/2$, and $1$ with prob $1/4$. Under signal $\pi'$, $\tilde{\omega} = 1/4$ with prob $1/2$ and $3/4$ with prob $1/2$. Sufficiency implies that, conditional on the belief $\tilde{\omega} = 1/4$ being realized after observing $\pi'$, the updated belief after observing $\pi$ should be $\tilde{\omega} = 0$ with prob $1/2$ and $\tilde{\omega} = 1/2$ with prob $1/2$. In contrast, Blackwell domination allows for different mappings. For example, conditional on $\tilde{\omega} = 1/4$ after observing $\pi'$, the belief after $\pi$ could be $\tilde{\omega} = 1/2$ with prob $1/2$ and $\tilde{\omega} = 1$ with prob $1/2$.

However, we can construct another signal, $\pi''$, by applying the dilation $K$ (as defined in footnote \ref{footnote:dilation}) to split $\pi'$. By construction, $\pi''$ is sufficient for $\pi'$, and $\pi'' \sim \pi$. }
If $\pi \succeq_{\theta} \pi'$, then there exists $\pi'' \sim_\theta \pi$ such that $\pi''$ is sufficient for $\pi'$ in terms of $\tilde{\theta}$, i.e., $\pi'' \sim_{\theta} \pi'' \lor \pi'$.\footnote{We can extend a Markov kernel $Q : \Delta(\Theta) \to \Delta(\Delta(\Theta))$ to a Markov kernel
$K : \Delta(\Omega) \to \Delta(\Delta(\Omega))$ by the following construction: for any
$\mu' \in \Delta(\Omega)$ with marginal $\mu'^{\theta} = \nu'$, define $K (\nu'' \otimes \mu'(\tilde{\omega}|\tilde{\theta}) |\mu') = Q(\nu''|\nu')$, where $\mu'(\tilde{\omega}|\tilde{\theta})$ is the conditionally distribution of $\tilde{\omega}$ given $\tilde{\theta}$ induced by $\mu'$. In other words, $K(\cdot | \mu')$ is defined as the pushforward of $Q(\cdot | \nu')$ under the map $\nu'' \mapsto \nu'' \otimes \mu'(\cdot | \tilde{\theta})$, which reconstructs a measure on $\Omega$ from its marginal $\nu''$ on $\Theta$ and the conditionals $\mu'(\cdot | \tilde{\theta})$.}

\begin{proof}[Proof of Lemma \ref{lemma:blackwell-frontier-P}]
    ``If''. Suppose $\pi \preceq \pi'$ where $\langle\pi'\rangle^\theta \in \overline{\mathcal{P}}$. According to Lemma~\ref{lemma:marginal_dominate}, $\langle\pi\rangle^\theta \preceq \langle \pi' \rangle^\theta$. Therefore by Assumption~\ref{ass:blackwell_closeness}, $\pi \in \Pi_\mathcal{P}$.
    
    ``Only if''. Suppose $\langle \pi \rangle^{\theta} \notin \overline{\mathcal{P}}$; otherwise the claim is trivial since $\pi \preceq \pi$. Then, there exists a signal $\pi'$ with $\langle \pi \rangle^{\theta} \in \overline{\mathcal{P}}$ such that $\pi'$ is sufficient for $\pi$ in terms of $\tilde{\theta}$. Then, since $\langle \pi \lor \pi'\rangle^{\theta} = \langle \pi' \rangle^{\theta} \in \mathcal{P}$, $\pi \lor \pi'$ is a $\mathcal{P}$-privacy-constrained signal. We get $\pi \preceq\pi \lor \pi'$.
\end{proof}

\begin{lemma}\label{preserving}
    For any signal $\pi$ with $\langle \pi \rangle^{\theta} \in \mathcal{P}$, $\pi \in \overline{\Pi}_{\mathcal{P}}$ if and only if there does not exist a conditionally (given $\pi$) privacy-preserving signal $\tilde{q}$ such that $\pi \prec \pi \lor \tilde{q}$.
\end{lemma}
\begin{proof}
    ``Only if''. Suppose there exists such signal $\tilde{q}$. Since $\tilde{q}$ is conditionally privacy-preserving given $\pi$, $\pi \lor \tilde{q}$ is $\mathcal{P}$-privacy-constrained. $\pi \prec \pi \lor \tilde{q}$ contradicts $\pi \in \overline{\Pi}_{\mathcal{P}}$.

    ``If''. Suppose there exists a $\mathcal{P}$-privacy-constrained signal $\pi'$ such that $\pi \prec \pi'$. Since $\langle \pi \rangle^{\theta} \in \mathcal{P}$, $\langle \pi' \rangle^{\theta} = \langle \pi \rangle^{\theta}$. W.l.o.g, assume $\pi'$ is  sufficient for $\pi$, i.e., $\pi \lor \pi' \sim \pi'$. Since $\langle \pi \lor \pi' \rangle^{\theta} = \langle \pi' \rangle^{\theta} = \langle \pi \rangle^{\theta}$, $\pi'$ is a conditionally (given $\pi$) privacy-preserving signal such that $\pi \prec \pi \lor \pi'$.
\end{proof}

\begin{proof}[Proof of Theorem \ref{thm:decomposition}]
    ``If''.
    Let $\tau_\gamma \in \underline{\mathcal{T}}_{\gamma}$ for some $\gamma \in \overline{\mathcal{P}}$
    and $\tilde{q}$ denote a signal which is Blackwell-undominated conditionally privacy-preserving given $\tau_\gamma$. Suppose $\tau_\gamma \lor \tilde{q}\notin \overline{\Pi}_\mathcal{P}$. Following Lemma \ref{preserving}, there exists another conditionally (given $\tau_\gamma \lor \tilde{q}$) privacy-preserving signal $\tilde{q}'$ such that $\tau_\gamma \lor \tilde{q} \prec \tau_\gamma \lor \tilde{q} \lor \tilde{q}'$. Since  $\langle \tau_\gamma \lor \tilde{q} \lor \tilde{q}' \rangle^{\theta} = \langle \tau_\gamma \lor \tilde{q}\rangle^\theta = \tau_\gamma^{\theta}$, $\tilde{q} \lor \tilde{q}'$ is also a conditionally (given $\tau_\gamma$) privacy-preserving signal. Therefore,  $\tau_\gamma 
    \lor \tilde{q} \prec \tau_\gamma \lor (\tilde{q} \lor \tilde{q}')$ contradicts that $\tilde{q}$ is  Blackwell-undominated among all conditionally privacy-preserving signals given $\tau_\gamma$.

    ``Only if''. 
    For any $\pi \in \overline{\Pi}_{\mathcal{P}}$, if $\pi \notin \underline{\mathcal{T}}_{\langle \pi \rangle^\theta}$, then there exists another signal $\pi_1 \in \Pi_{\mathcal{P}}$ such that $\langle \pi_1 \rangle^\theta = \langle \pi \rangle^\theta$ and $\pi$ is sufficient of $\pi_1$. If $\pi_1 \notin \underline{\mathcal{T}}_{\langle \pi \rangle^\theta}$, we can find another signal $\pi_2 \in \Pi_{\mathcal{P}}$ such that $\langle \pi_2 \rangle^\theta = \langle \pi \rangle^\theta$ and $\pi_1$ is sufficient of $\pi_2$. Continuing this process, if it terminates in a finite number of steps, we will eventually obtain a signal $\pi_N \in \underline{\mathcal{T}}_{\langle \pi \rangle^\theta}$. Otherwise, we construct a sequence $\{\tilde{\mu}_{\pi_t}\}_{t \in \mathbb{N}_+}$, where $\tilde{\mu}_{\pi_t}$ is random variable about posterior belief induced by $\pi_t$. According to Section 4 in \cite{blackwell1953equivalent}, the sequence $\{\tilde{\mu}_{\pi_t}\}_{t \in \mathbb{N}}$ forms a reverse martingale. 
    By the martingale convergence theorem \citep{doob1951continuous}, $\tilde{\mu}_{\pi_t} \to \tilde{\mu}^*$ as $t \to \infty$. Let $\pi_{\infty}$ be one of the signals that induces $\tilde{\mu}^*$. Then, we have $\pi_{\infty} \in \underline{\mathcal{T}}_{\langle \pi \rangle^\theta}$. Define $\pi^*$ as a unified notation that refers to either $\pi_N$ or $\pi_\infty$. Since $\pi$ is  conditionally privacy-preserving signal given $\pi^*$. Therefore, there is a $\langle\pi\rangle^\theta \in \mathcal{\overline{P}}$ (Lemma~\ref{lemma:blackwell-frontier-P}), $\pi^* \in \mathcal{T}$ and $\pi$ is a corresponding Blackwell-undominated conditional privacy-preserving signal such that $\pi \sim \pi^* \lor \pi$.
\end{proof}

\begin{proof}[Proof of Theorem \ref{prp:least-informative-signal}]
    ``Only if''. If condition~\eqref{least.1} is not satisfied, by definition $\underline{\tau} \notin \underline{\mathcal{T}}_{\gamma}$. Now, suppose condition~\eqref{least.2} is not satisfied. Then, there exists a positive measure set $B \in \operatorname{supp(\underline{\tau}^{\theta}})$ such that for any $\mu \in \operatorname{supp}(\underline{\tau})$ with $\mu^{\theta} \in B$, $Q_{\underline{\tau}}(\mu|\mu^\theta) < 1$, where $Q_{\underline{\tau}}$ denotes the conditional distribution over $\Delta(\Omega)$ given $\Delta(\Theta)$ induced by $\underline{\tau}$. By construction, $Q_{\underline{\tau}}$ is non-degenerate. We can then construct a new distribution over posteriors, $\tau'$, such that $\tau'^{\theta} = \underline{\tau}^{\theta}$ and for all $\mu^{\theta} \in \operatorname{supp}(\underline{\tau}^{\theta})$, in the new distribution $\tau'$, the conditional probability on $\hat{\mu} = \mathbb{E}_{\underline{\tau}}(\{\mu \in \Delta(\Omega):\mu^\theta =\nu\}|\nu)$ is $1$ given $\nu \in \operatorname{supp} (\underline{\tau}^\theta)$. $\underline{\tau}$ is a strict mean-preserving spread of $\tau'$, since $\underline{\tau}(\cdot) = \int_{\Delta (\Omega)} Q_{\underline{\tau}}(\cdot|\mu'^{\theta}) d \tau'(\mu')$, and $Q_{\underline{\tau}}(\cdot|\mu^{\theta})$ is non-degenerate. $\tau' \prec \underline{\tau}$ contradicts $\underline{\tau} \in \underline{\mathcal{T}}_{\gamma}$.

    ``If''. Suppose there exists a signal $\tau'$ such that $\tau'^{\theta} \in \overline{\mathcal{P}}$ and $\tau' \prec \underline{\tau}$. There is a non-degenerate dilatation $K:\Delta(\Omega) \to \Delta(\Delta(\Omega))$ spreads $\tau'$ to $\underline{\tau}$. Moreover, since $\underline{\tau}$ satisfies condition~\ref{least.2}, there is a one-to-one mapping from $\nu \in \operatorname{supp}(\underline{\tau}^\theta)$ to $\mu \in \operatorname{supp}(\underline{\tau})$. Therefore, the kernel $Q:\Delta(\Theta) \to \Delta(\Delta(\Theta))$ defined in Lemma~\ref{lemma:marginal_dominate} is non-degenerate if  $K$ is non-degenerate. As a result, $\tau'^{\theta} \prec \underline{\tau}^{\theta}$, which contradicts $\tau'^{\theta} \in \overline{\mathcal{P}}$.
\end{proof}

\subsection{Proofs for Section~\ref{sec:privacy_permissible_set}}

\begin{proof}[Proof of Proposition \ref{prp:extreme_points}]
    ``If''. Suppose there exists another $\gamma' \in \mathcal{P}_{\mathcal{M}}$ such that $\gamma \prec \gamma'$, then there exists a nondegenerate dilation $K: \Delta(\Theta) \to \Delta(\Delta(\Theta))$ such that for almost every $\nu \in \operatorname{supp}(\gamma)$, $$\nu = \int_{\Delta(\Theta)} \nu' K(d\nu' |\nu).$$
    This means that for almost every $\nu \in \operatorname{supp}(\gamma)$, it can be expressed by linear combination of $\nu' \in \operatorname{supp}(\gamma')$.
    Since $K$ is nondegenerate, there exists a positive measure subset $A \subseteq \operatorname{supp}(\gamma)$ such that $K(\nu|\nu) < 1$ for any $\nu \in A$. Hence,  $\nu \notin \operatorname{ext}\mathcal{M}$ for any $\nu \in A$, which is contradicted with $\gamma(\operatorname{ext}\mathcal{M}) = 1$.

    ``Only if''. We construct a dilation from $\mathcal{M}$ to $\Delta(\operatorname{ext} \mathcal{M})$. Since $(\Theta, \mathcal{B}(\Theta))$ is a standard Borel space, $\Delta(\Theta)$ embeds into a locally convex space and endowed with the topology of weak convergence is metrizable. $\mathcal{M}$ is a compact convex subset of $\Delta(\Theta)$, which  is also metrizable. By Choquet’s Theorem (Theorem 10.7, p.168, \cite{simon2011convexity}), for any $\nu \in \mathcal{M}$, the set
    $$\Phi(\nu) := \Bigl\{ P_\nu \in \Delta(\operatorname{ext}\mathcal{M}) : \nu = \int \nu'  dP_\nu(\nu') \Bigr\}$$
    is nonempty. Moreover, $\Phi(\mu)$ is closed in the weak-* topology. Define the barycenter map $B : \Delta(\mathcal{M}) \to \mathcal{M}$ by $B(P_\nu) = \int \nu'  dP_\nu = \nu$. By \cite{simon2011convexity}, Theorem 9.1 (p.136), the map $B$ is continuous. Consequently, for any open set $U \subseteq \Delta(\operatorname{ext}\mathcal{M})$, $\Phi^{-1}(U) = \{\nu \in \mathcal{M} : B^{-1}({\nu}) \cap U \neq \emptyset\} = B(U)$,
    which is an open set. Therefore, by the Kuratowski–Ryll-Nardzewski measurable selection theorem (Theorem 6.9.3, p.36, Vol. II, \cite{bogachev2007measure}), there exists a measurable selection $P_\nu^* : \mathcal{M} \to \Delta(\operatorname{ext}\mathcal{M})$ such that $\nu = \int \nu'  dP_\nu^*(\nu')$ Hence, the map $K : \nu \mapsto P_\nu^*$ defines a dilation. If $\gamma(\operatorname{ext}\mathcal{M}) \neq 1$, then $K$ is nondegenerate, which implies $\gamma \notin \overline{\mathcal{P}}_{\mathcal{M}}$.
\end{proof}

\begin{lemma}
    Suppose $\gamma \in \mathcal{P}_\mathcal{I}$. For almost every $\nu \in \operatorname{supp}(\gamma)$, there is a nonnegative measurable function $g_\nu$ defined for which
    \begin{equation}\label{Radon-dev}
        \nu(B) = \int_{B} g_\nu d\mu_0^\theta, \quad  \text{for all } B \in \mathcal{B}(\Theta).
    \end{equation}
    and $g_\nu$ is essential bounded. Let $e^{\overline{\varepsilon}_\nu}$ as the essential supremum of $g_\nu$, i.e,
    \begin{equation}\label{sup}
        \begin{aligned}
            \mu_0^\theta\left(\left\{\theta \in \Theta:g_\nu(\theta) > e^{\overline{\varepsilon}_\nu}\right\}\right) &= 0, \\
            \mu_0^\theta\left(\left\{\theta \in \Theta:g_\nu(\theta) > e^{\overline{\varepsilon}_\nu} - \delta \right\}\right) &> 0, \quad \text{for all } \delta > 0.
        \end{aligned}
    \end{equation}
    then, $\overline{\varepsilon}_\nu < \varepsilon$ and
    \begin{equation}\label{essentialbounded}
        \mu_0^\theta\left(\left\{\theta \in \Theta:g_\nu(\theta) \in[e^{\overline{\varepsilon}_\nu - \varepsilon}, e^{\overline{\varepsilon}_\nu}]\right\}\right) = 1.
    \end{equation}
    Moreover, if $g_\nu \neq 1$ $\mu_0^\theta$-almost surely, then $\overline{\varepsilon}_\nu > 0$ and
    \begin{equation}\label{split-half}
        \begin{aligned}
            \mu_0^\theta \left(\left\{\theta \in \Theta:g_\nu(\theta) \in[e^{\overline{\varepsilon}_\nu-\varepsilon}, 1)\right\}\right) > 0,\\
            \mu_0^\theta\left(\left\{\theta \in \Theta: g_\nu(\theta) \in (1, e^{\overline{\varepsilon}_\nu}]\right\}\right) > 0.
        \end{aligned}
    \end{equation}
\end{lemma}
\begin{proof}
    Suppose $\gamma \in \mathcal{P}_\mathcal{I}$. Radon-Nikodym Theorem (\cite{royden2010real}, p.382) shows that for almost every $\nu \in \operatorname{supp}(\gamma)$, $g_\nu$ defined by (\ref{Radon-dev}) exists. By the inferential-privacy constraint (\ref{eq:ip}), for all $B \in \mathcal{B}(\Theta)$ with $\mu^\theta_0(B) > 0$,
    \begin{equation*}
        e^{-\varepsilon} \frac{\mu_0^\theta(B)}{\mu_0^\theta(\Theta)} \leq\frac{\nu(B)}{\nu(\Theta)} \leq e^{\varepsilon} \frac{\mu^\theta_0(B)}{\mu^\theta_0(\Theta)} \Rightarrow e^{-\varepsilon} \mu_0^\theta(B) \leq \nu(B) \leq  e^\varepsilon \mu_0^\theta(B),
    \end{equation*}
    which is due to the fact that $\mu_0^\theta(\Theta) = \nu(\Theta) = 1$. Therefore,
    \begin{equation*}
        \mu_0^\theta \left(\left\{\theta \in \Theta: g_\nu(\theta) \in[e^{-\varepsilon}, e^\varepsilon]\right\}\right) = 1.
    \end{equation*}
    $g_\mu$ is essentially bounded blow and above. Since the completeness of $\mathbb{R}$, $g_\nu$ has essential supremum and infimum, denoted as $e^{\overline{\varepsilon}_\nu}$ and $e^{\underline{\varepsilon}_\nu}$, respectively.

    Since the fact that $\nu(\Theta) = \nu_0^\theta(\Theta)$, if
    \begin{equation*}
        \begin{aligned}
            \mu_0^\theta\left(\left\{\theta \in \Theta:g_\nu(\theta) \in[e^{\underline{\varepsilon}_\nu}, 1)\right\}\right) > 0, \text{ then } \\
            \mu_0^\theta\left(\left\{\theta \in \Theta: g_\nu(\theta) \in (1, e^{\overline{\varepsilon}_\nu}]\right\}\right) > 0. \quad 
        \end{aligned}
    \end{equation*}
    Hence, if $g_\nu \neq 1$ $\mu_0^\theta$-almost surely, (\ref{split-half}) holds and $\underline{\varepsilon}_\nu < 0 <\overline{\varepsilon}_\nu$. We need only to show that $\overline{\varepsilon}_\nu < \varepsilon$ and $\overline{\varepsilon}_\nu - \underline{\varepsilon}_\nu \leq   \varepsilon$, which induce (\ref{essentialbounded}).

    Suppose $\overline{\varepsilon}_\nu - \underline{\varepsilon}_\nu > \varepsilon$, then for a constant number $0 < \Delta \varepsilon < \frac{1}{2} (\overline{\varepsilon}_\nu - \underline{\varepsilon}_\nu - \varepsilon)$, we have  $(\overline{\varepsilon}_\nu - \Delta\varepsilon) - (\underline{\varepsilon}_\nu + \Delta \varepsilon) > \varepsilon$. Since the definition of essential supremum and infimum, define two measurable sets
    \begin{equation*}
        B_1 := \left\{\theta \in \Theta: g_\nu(\theta) < e^{\underline{\varepsilon}_\nu + \Delta \varepsilon}\right\},
    \end{equation*}
    \begin{equation*}
        B_2 := \left\{\theta \in \Theta: g_\nu(\theta) > e^{\overline{\varepsilon}_\nu - \Delta \varepsilon}\right\},
    \end{equation*}
    $\mu^\theta_0(B_1) > 0$ and $\mu_0^\theta(B_2) > 0$. Obviously, 
    \begin{equation*}
        \frac{\nu(B_2)}{\nu(B_1)} > e^{\overline{\varepsilon}_\nu - \Delta\varepsilon - (\underline{\varepsilon}_\nu+ \Delta \varepsilon)} \frac{\mu_0^\theta(B_2)}{\mu_0^\theta(B_1)} > e^\varepsilon \frac{\mu_0^\theta(B_2)}{\mu_0^\theta(B_1)},
    \end{equation*}
    contradicts with (\ref{eq:ip}).

    If $g_\nu = 1$ $\mu_0^\theta$-almost surely, then $\overline{\varepsilon}_\nu = 0 < \varepsilon$. Otherwise, since above we show that $\underline{\varepsilon}_\nu < 0 < \overline{\varepsilon}_\nu$ and $\overline{\varepsilon}_\nu - \underline{\varepsilon}_\nu \leq \varepsilon$, then $\overline{\varepsilon}_\nu < \varepsilon$.
\end{proof}

\begin{proof}[Proof of Proposition~\ref{prp:ip-froniter}]
    ``If''. We will show that for almost every $\nu \in \operatorname{supp}(\gamma)$, there does not exist non-degenerated $K_\nu \in \Delta(\Delta(\Theta))$ such that 
    $\nu = \int_{\Delta(\Theta)} \nu' K_\nu(d\nu')$ and inferential-privacy constraint (\ref{eq:ip}) holds almost everywhere on $\operatorname{supp}(K_\nu)$. This statement indicates that $\gamma \in \overline{\mathcal{P}}_\mathcal{I}$.

    Suppose there is a non-degenerated $K_\nu$. Since $K_\nu$ is non-degenerated, there is a positive measurable subset $M \in \operatorname{supp}(K_\nu)$ and, w.l.o.g, a subset $B \subseteq E_\nu$ such that for all $\nu' \in M$, $\nu'(B) > \nu(B)$. Because of the inferential-privacy constraint (\ref{eq:ip}), $\frac{\nu'(B)}{\nu'(E_\nu^c)} \leq \frac{\nu(B)}{\nu(E_\nu^c)} = e^\varepsilon \frac{\mu_0^\theta(B)}{\mu_0^\theta(E_\nu^c)}$, then $\nu'(E_\nu^c) > \nu(E_\nu^c)$. Since the mean-preserving condition, there is another positive measurable subset $M' \in \operatorname{supp}(K_\mu)$ such that $\nu''(E_\nu^c) < \nu(E_\nu^c)$ holds for all $\nu'' \in M'$. Again, due to the constraint (\ref{eq:ip}), $\frac{\nu''(E_\nu)}{\nu''(E_\nu^c)} \leq \frac{\nu(E_\nu)}{\nu(E_\nu^c)} = e^\varepsilon\frac{\mu_0^\theta(E_\nu)}{\mu_0^\theta(E_\nu^c)}$, then $\nu''(E_\nu) < \nu(E_\nu)$. Therefore $\nu''(\Theta) = \nu''(E_\nu) + \nu''(E_\nu^c) < \nu(E_\nu) + \nu(E_\nu^c) = 1$ which contradicts with the fact that $\nu''$ is a probability.

    ``Only if''. Suppose a $\gamma'$ for which there is a positive measurable subset $M \in \operatorname{supp}(\gamma')$ and for each $\nu \in M$, there is a positive measurable subset $F \subseteq \Theta$ and some $\varepsilon' \in(0, \varepsilon)$ such that for any $B_1 \in \mathcal{B}(F)$ and $B_2 \in \mathcal{B}(\Theta)$ with $\mu_0^\theta(B_1) > 0$ and $\mu_0^\theta(B_2) > 0$,  
    \begin{equation}\label{bound}
        e^{-\varepsilon'}\frac{\mu_0^\theta(B_1)}{\mu_0^\theta(B_2)} \leq \frac{\nu(B_1)}{\nu(B_2)} \leq e^{\varepsilon'}\frac{\mu_0^\theta(B_1)}{\mu_0^\theta(B_2)}\quad  \mu_0^\theta\text{-almost surely}.
        \end{equation} 
    Then, for a positive constant $\delta < \min\{ e^{\varepsilon - \varepsilon'} - 1, 1/\nu(F)\}$, $\nu$ can split into $\nu_1$ with probability $\frac{1}{2}(1+\delta\mu(F))$ and $\nu_2$ with probability $\frac{1}{2}(1-\delta \nu(F))$, where $\nu_1(\cdot) :=  \frac{(1+\delta) \nu(\cdot\cap F)}{1+\delta \nu(F)} + \frac{\nu (\cdot \setminus F)}{1+\delta \nu(F)}$, $\nu_2(\cdot) := \frac{(1-\delta)\nu(\cdot \cap F)}{1-\delta \nu(F)} + \frac{\nu(\cdot \setminus F)}{1-\delta\nu(F)}$. Since for $B_3, B_4 \in \mathcal{B}(\Theta)$ with $\mu_0^\theta(B_4) >0$, almost surely, 
    \begin{equation*}
        \begin{aligned}
            \frac{\nu_1(B_3)}{\nu_1(B_4)} & = \frac{\nu_1(B_3 \cap F)}{\nu_1(B_4)} + \frac{\nu_1(B_3 \setminus F)}{\nu_1(B_4)} 
            \leq \frac{\nu(B_3 \cap F)(1+\delta)}{\nu(B_4)} + \frac{\nu(B_3\setminus F)}{\nu(B_4)} \\
            & \leq (1+\delta)e^{\varepsilon'} \frac{\nu_0^\theta(B_3\cap F)}{\nu_0^\theta(B_4)} + e^\varepsilon \frac{\nu_0^\theta(B_3 \setminus F)}{\nu_0^\theta (B_4)} 
            \leq e^\varepsilon \frac{\nu_0^\theta(B_3)}{\mu_0^\theta (B_4)},\\
            \frac{\nu_2(B_3)}{\nu_2(B_4)} & = \frac{\nu_2(B_3 \cap F)}{\nu_2(B_4)} + \frac{\nu_2(B_3 \setminus F)}{\nu_2(B_4)}
            \geq \frac{\nu(B_3 \cap F)(1- \delta)}{\nu(B_4)} + \frac{\nu(B_3 \setminus F)}{\nu(B_4)} \\
            & \geq (1-\delta)e^{-\varepsilon'} \frac{\mu_0^\theta(B_3\cap F)}{\mu_0^\theta (B_4)} + e^{-\varepsilon} \frac{\mu_0^\theta (B_3 \setminus F)}{\mu_0^\theta (B_4)}
            \geq e^{-\varepsilon} \frac{\mu_0^\theta(B_3)}{\mu_0^\theta (B_4)},
        \end{aligned}
    \end{equation*}
    $\nu_1$ and $\nu_2$ satisfies $\varepsilon$-inferentially private constraint (\ref{eq:ip}). Thus, $\gamma' \notin \overline{\mathcal{P}}_\mathcal{I}$.

    Next, we only need to show that if $\gamma'$ does not satisfies (\ref{prp_Dichotomy}), $\gamma'$ has a positive measurable $M \in \operatorname{supp}(\gamma')$ and for each $\mu \in M$, there is a subset $F \in \mathcal{B}(\Theta)$ with $\mu_0^\theta(F) > 0$ and some $\varepsilon' \in (0,\varepsilon)$ such that \eqref{bound} holds. 
        
    Let $M$ be the set of $\nu \in \operatorname{supp}(\gamma')$ such that (\ref{prp_Dichotomy}) does not hold and there is a measurable function $g_\nu$ defined by (\ref{Radon-dev}). Suppose $g_\nu \neq 1$ $\mu_0^\theta$-almost surely, otherwise $\nu = \mu_0^\theta$ almost surely which is a trivial case. Then, $\overline{\varepsilon}_\nu$ defined as (\ref{sup}) is contained in $(0,\varepsilon)$. Define two measurable sets 
    \begin{equation*}
        \begin{aligned}
            \overline{B} &:=\left\{\theta \in \Theta : g_\nu (\theta) = e^{\overline{\varepsilon}_\nu} \right\}, \\
            \underline{B} &:=\left\{\theta \in \Theta : g_\nu (\theta) = e^{\overline{\varepsilon}_\nu-\varepsilon} \right\}.
        \end{aligned}
    \end{equation*}
        
    \begin{enumerate}[(1)]
        \item Suppose $\mu_0^\theta(\overline{B}) = 0$ or $\mu_0^\theta(\underline{B}) = 0$, and w.l.o.g. assume $\mu_0^\theta(\overline{B}) = 0$, that is, 
        \begin{equation*}
            \mu_0^\theta\left(\left\{\theta \in \Theta:g_\nu(\theta) \in[e^{\overline{\varepsilon}_\nu- \varepsilon}, e^{\overline{\varepsilon}_\nu})\right\}\right) = 1.
        \end{equation*}
        Similarly as (\ref{split-half}), we can show that
        \begin{equation}\label{one-weak}
            \mu_0^\theta \left(\left\{\theta \in \Theta: g_\nu(\theta) \in [1, e^{\overline{\varepsilon}_\nu})\right\}\right) > 0.
        \end{equation}
        
        \item Suppose $\mu_0^\theta(\overline{B}) > 0$ and $\mu_0^\theta(\underline{B}) > 0$. If $\mu_0^\theta(\overline{B} \cup \underline{B}) = 1$, then $\nu$ satisfies (\ref{prp_Dichotomy}). Hence, under the assumption that $\nu$ does not satisfy (\ref{prp_Dichotomy}), 
        \begin{equation*}
            \mu_0^\theta \left(\left\{\theta \in \Theta: g_\nu(\theta) \in(e^{\overline{\varepsilon}_\nu- \varepsilon}, e^{\overline{\varepsilon}_\nu})\right\}\right) > 0.
            \end{equation*}
    \end{enumerate}
    W.l.o.g., we can assume that (\ref{one-weak}) holds. Since 
    \begin{equation*}
        \begin{aligned}
            0 & = \mu_0^\theta\left(\left\{\theta \in \Theta: g_\nu(\theta) \in [e^{\overline{\varepsilon}_\nu}, e^{\overline{\varepsilon}_\nu})\right\}\right)\\
            & = \mu_0^\theta \left( \bigcap_{n=1}^{\infty} \left\{\theta \in \Theta: g_\nu(\theta) \in (e^{\frac{n-1}{n}\overline{\varepsilon}_\nu}, e^{\overline{\varepsilon}_\nu})\right\}\right), 
        \end{aligned}
    \end{equation*}
    where the first equality is due to $\left\{\theta \in \Theta: g_\nu(\theta) \in [e^{\overline{\varepsilon}_\nu}, e^{\overline{\varepsilon}_\nu})\right\} = \emptyset$ and the second equality holds since for any $\theta \in \Theta$ for which $g_\nu(\theta) < e^{\overline{\varepsilon}_\nu}$, there is $N > 0$ such that $g_\nu(\theta) < \frac{N-1}{N} e^{\overline{\varepsilon}_\nu}$. Therefore, since the continuity of (finite) measure, there exists $N > 0$ such that $$\mu_0^\theta \left( \bigcap_{n=1}^{N} \left\{\theta \in \Theta: g_\nu(\theta) \in (e^{\frac{n-1}{n}\overline{\varepsilon}_\nu}, e^{\overline{\varepsilon}_\nu})\right\}\right) < \mu_0^\theta \left(\left\{\theta \in \Theta: g_\nu(\theta) \in [1, e^{\overline{\varepsilon}_\nu})\right\}\right),$$ 
    then
    \begin{equation}\label{closed-bound}
        \mu_0^\theta \left(\left\{\theta \in \Theta: g_\nu(\theta) \in [1, e^{\frac{N-1}{N}\overline{\varepsilon}_\nu}]\right\}\right) > 0.
    \end{equation}

    Denote $F:= \left\{\theta \in \Theta: g_\nu(\theta) \in [1, e^{\frac{N-1}{N}\overline{\varepsilon}_\nu}]\right\}$, then since (\ref{closed-bound}) and (\ref{essentialbounded}), for any $B_1 \in \mathcal{B}(F)$ and $B_2 \in \mathcal{B}(\Theta)$ with $\mu_0(B_1) > 0$ and $\mu_0(B_2) > 0$,
    \begin{equation*}
        e^{- \overline{\varepsilon}_\nu} \frac{\mu_0^\theta (B_1)}{\mu_0^\theta (B_2)} \leq \frac{\nu(B_1)}{\nu(B_2)} \leq e^{\varepsilon - \frac{1}{N}\overline{\varepsilon}_\nu} \frac{\mu_0^\theta (B_1)}{\mu_0^\theta (B_2)}.
        \end{equation*} 
    Since $\overline{\varepsilon}_\nu \in (0, \varepsilon)$, we construct a set $F$ satisfying (\ref{bound}).
\end{proof}

\begin{proof}[Proof of Proposition \ref{prp:expectation}]
    ``Only if''. We first show the part (1).
    \begin{enumerate}[(i)]
        \item Suppose there is a subset $Y \in \tilde{f}(\Theta)$ which contains more than two points such that $\nu$ puts a positive probability on $\tilde{f}^{-1}(y)$ for all $y \in Y$. Since $\mathbb{E}_\nu[\tilde{f}(\theta)]$ is a combination of $\tilde f(\theta)$ for $\theta \in \operatorname{supp}(\nu)$, there exists $\underline{y} < \mathbb{E}_\nu[\tilde{f}(\theta)] < \overline{y}$ such that $\underline{p} := \nu(\{\theta \in \Theta | \tilde{f}(\theta) \leq \underline{y}\}) > 0$, $\overline{p} := \nu(\{\theta \in \Theta | \tilde{f}(\theta) \geq \overline{y}\}) > 0$, $\underline{p}+\overline{p} < 1$ and there is $\alpha \in (0,1)$ such that $\alpha \mathbb{E}_\nu[\tilde{f}(\theta)| \tilde{f}(\theta) \leq \underline{y}] + (1-\alpha) \mathbb{E}_\nu[\tilde{f}(\theta) | \tilde{f}(\theta) \geq \overline{y}] = \mathbb{E}_\nu [\tilde{f}(\theta)].$ Denote $\lambda := \min\{\frac{\underline{p}}{\alpha}, \frac{\overline{p}}{1-\alpha}\} < 1$. Then $\nu$ can be split into $\nu_1(\theta):= \mathbf{1}\{\tilde f(\theta) \leq \underline{y}\}\alpha \frac{\nu(\theta)}{\underline{p}} + \mathbf{1}\{\tilde f(\theta) \geq \overline{y}\} (1-\alpha) \frac{\nu(\theta)}{\overline{p}}$ and $\nu_2 = \frac{1}{1-\lambda}(\nu-\lambda \nu_1)$ with probability $\lambda$ and $(1-\lambda)$ respectively. Since $\mathbb{E}_{\nu_1}[\tilde{f}(\theta)] = \mathbb{E}_\nu[\tilde{f}(\theta)] = \mathbb{E}_{\nu_2}[\tilde{f}(\theta)]$, this split does not change the distribution of posterior mean of $\tilde{f}$.
        \item Suppose there a positive measure subset of belief such that $\nu$ puts $\alpha$ on the set $\tilde{f}^{-1}(y_1)$ and $(1-\alpha)$ on the set $\tilde{f}^{-1}(y_2)$ but $\operatorname{supp}(\nu) \cap \tilde{f}^{-1}(y_1)$ is not a singleton. Denote $\nu_\theta(\theta') := \alpha \mathbf{1}\{\theta' = \theta\} + \mathbf{1}\{\theta' \in \tilde f^{-1}(y_2)\} \nu(\theta')$ for $\theta \in \operatorname{supp}(\nu) \cap \tilde{f}^{-1}(y_1)$, then $\nu = \int_{ \operatorname{supp}(\nu) \cap \tilde f^{-1}(y_1)} \nu_\theta d\nu(\theta)$. Since such $\mu$ has a split, and notice that $\mathbb{E}_{\nu_\theta} [\tilde f(\theta)] = \mathbb{E}_{\nu} [\tilde f(\theta)]$ for all $\theta \in \operatorname{supp}(\nu) \cap \tilde f^{-1}(y_1)$, $\gamma$ has a strictly mean-preserving spreading that cannot change the induced distribution of posterior mean of $\tilde{f}$.
    \end{enumerate}

     Suppose the part (2) does not hold, then $\kappa_\gamma$ is a strictly mean-preserving contraction of $\overline{\kappa}$, i.e., there exists a non-degenerated dilation $K: \mathbb{R} \to \Delta (\mathbb{R})$, that is, $y = \int_{\mathbb{R}} y' K(dy'|y)$ for all $y \in \operatorname{supp}(\kappa_\gamma),$ such that 
    \begin{equation*}
        \overline{\kappa}(B) = \int_{\mathbb{R}} K(B|y) d\kappa_\gamma, \quad \forall B\in \mathcal{B}(\mathbb{R}).
    \end{equation*}    
    That means, there is a positive measure subset $Y \subseteq \operatorname{supp}(\kappa_\gamma)$ such that any $y \in Y$ is split into $\operatorname{supp}(K(\cdot|y))$ according to non-degenerated distribution measure $K(\cdot|y)$. We now construct a dilation form $\Delta(\Theta)$ to $\Delta(\Delta(\Theta))$ based on $K$.

    W.l.o.g., suppose the part (1) holds. Denote $E$ as the set of $y \in \operatorname{supp}(\kappa_\gamma)$ such that almost surely, all $\nu \in \operatorname{supp}(\gamma)$ with $\mathbb{E}_\nu [\tilde{f} (\theta)] = y$ satisfies $\operatorname{supp}(\nu) \subseteq \tilde{f}^{-1}(y)$. $\operatorname{supp}(\kappa_\gamma) \setminus E$ has a positive measure. Otherwise,  by the definition of $E$, $\kappa_\gamma$ disclose full information about posterior mean of $\tilde{f}$. Hence $\kappa_\gamma = \overline{\kappa}$, contradiction. Assume $Y \subseteq \operatorname{supp}(\kappa_\gamma) \setminus E$. For $y \in Y$, it is comprised by some $\nu \in \operatorname{supp}(\gamma)$ such that its support consists of two points $\theta_1 \in \tilde{f}^{-1}(y_1)$ and $\theta_2 \in \tilde{f}^{-1}(y_2)$ , where $y_1 < y_2$. For each $y = \alpha y_1 + (1-\alpha) y_2$, with $\alpha \in [0,1]$, define $\nu_y := \alpha \delta_{\theta_1} + (1-\alpha) \delta_{\theta_2}$.
    If $\operatorname{supp}(K(\cdot|y)) \subseteq [y_1, y_2]$, then $\nu = \int_{\operatorname{supp}(K(\cdot|y))} \nu_{y'} K(y'|y)$. Thus, $\gamma$ has a strictly mean-preserving spread and satisfies the constraint. If $\operatorname{supp}(K(\cdot|x)) \nsubseteq [y_1, y_2]$, then there exits a non-degenerated $Q(\cdot|y)$ which is the mean-preserving contraction of $K(\cdot|y)$ and $\operatorname{supp}(Q(\cdot|y)) \subseteq [y_1, y_2]$. Following the same argument above, replacing by $Q(\cdot|y)$.

    ``If''. Suppose there exists another $\gamma' \in \mathcal{P}_\mathcal{E}$ that strictly Blackwell-dominates $\gamma$, then there is a non-degenerated dilation $K: \Delta (\Theta) \to \Delta (\Delta (\Theta))$, that is, $\nu = \int_{\Delta (\Theta)} \nu' K(d\nu'|\nu)$, for all $\nu \in \operatorname{supp}(\gamma)$, such that $\gamma'(B) = \int_{\Delta(\Theta)} K(B|\nu) d\gamma,$ for all $B \in \mathcal{B}(\Delta(\Theta))$. Define $Q(Y|y):= \mathbb{E}_\gamma[K(\{\nu':\mathbb{E}_{\nu'}[\tilde{f}(\theta)] \in Y\}|\nu_y) \mid \{\nu_y:\mathbb{E}_{\nu_y}[\tilde{f}(\theta)] = y\}]$, for almost every $Y \in \mathcal{B}(\mathbb{R})$. Because of the part (2) that support set of all $\nu \in \operatorname{supp}(\gamma)$ contains at most two points, for all $\nu_y$ such that $K(\nu_y|\nu_y) < 1$, it must have $K(\{\nu:\mathbb{E}_{\nu}[\tilde{f}(\theta)] \neq y\}|\nu_y) > 0$. Because $K$ is non-degenerated, $\{\nu: K(\nu|\nu) < 1\}$ has a positive measure. Hence, $Q$ is non-degenerated. 
    
    The following proof is similar with proof of Lemma~\ref{lemma:marginal_dominate}.
    \begin{equation*}
        \begin{aligned}
            \kappa_{\gamma'}(Y) = & \gamma'(\{\nu: \mathbb{E}_{\nu}[\tilde{f}(\theta)] \in Y\}) 
            =  \int_{\Delta(\Theta)} K(\{\nu': \mathbb{E}_{\nu'}[\tilde{f}(\theta)] \in Y\} | \nu) d\gamma(\nu) \\
            = & \int_{\mathbb{R}} \int_{\{\nu_y: \mathbb{E}_{\nu_y}[\tilde{f}(\theta)] = y\}} K(\{\nu': \mathbb{E}_{\nu'}[\tilde{f}(\theta)] \in Y\} | \nu_y) d\gamma(\nu_y) dy
            =  \int_{\mathbb{R}} Q(Y|y) d\kappa_\gamma(y),
        \end{aligned}
    \end{equation*}
    \begin{equation*}
        \begin{aligned}
            y d\kappa_\gamma(y) =& \int_{\{\nu_y: \mathbb{E}_{\nu_y}[\tilde{f}(\theta)] = x\}} y d\gamma(\nu_y) 
            = 
            \int_{\{\nu_y: \mathbb{E}_{\nu_y}[\tilde{f}(\theta)] = y\}} \left[\int_{\Theta} \tilde{f}(\theta) d\left(\int_{\Delta(\Theta)} \nu' K(d\nu' |\nu_y)\right) \right] d\gamma(\nu_y) \\
            = & \int_{\{\nu_y: \mathbb{E}_{\nu_y}[\tilde{f}(\theta)] = y\}} \int_{\Delta(\Theta)}\int_{\Theta} \tilde{f}(\theta) d\nu' K(d\nu'|\nu_y) d\gamma(\nu_y)\\
            = & \int_{\{\nu_y: \mathbb{E}_{\nu_y}[\tilde{f}(\theta)] = y\}} \int_{\Delta(\Theta)} \mathbb{E}_{\nu'} [\tilde{f}(\theta)] K(d\nu'|\nu_y) d\gamma(\nu_y)\\
            = & \int_{\{\nu_y: \mathbb{E}_{\nu_y}[\tilde{f}(\theta)] = y\}} \int_{\mathbb{R}} y' dK(\{\nu_{y'}: \mathbb{E}_{\nu_{y'}} [\tilde{f}(\theta)] = y'\}|\mu_y) d\gamma(\nu_y) \\
            = & \int_{\mathbb{R}} y' Q(dy'|y) d\kappa_\gamma(y)
            \Rightarrow y = \int_{\mathbb{R}} y' Q(dy'|y) \text{ almost surely}.
        \end{aligned}
    \end{equation*}
    Thus, $\kappa_{\gamma'}$ is a strictly mean-preserving spread of $\kappa_\gamma = \overline{\kappa}$, which contradict with $\gamma' \in \mathcal{P}_\mathcal{E}$.
\end{proof}

\end{document}